\documentclass[journal,10pt,twoside]{IEEEtran}

\ifCLASSINFOpdf
\usepackage[pdftex]{graphicx}
\else
\fi
\usepackage{mathtools}
\usepackage{cite}
\usepackage[justification=centering]{caption}
\usepackage{graphicx}
\usepackage{textcomp}
\usepackage{xcolor}
\usepackage{booktabs}
\usepackage{epstopdf}
\usepackage{amsmath,amsthm, amssymb}
\usepackage{threeparttable}
\usepackage{bm}
\usepackage{multirow}
\usepackage{booktabs}
\usepackage{color}
\usepackage{tabularx}
\usepackage{setspace}
\usepackage{verbatim}
\usepackage{stfloats}
\usepackage{caption}
\usepackage{float}
\usepackage{url}
\usepackage{subfigure}

\usepackage{stfloats}
\usepackage{amsmath}
\allowdisplaybreaks[0]

\makeatletter
\newif\if@restonecol
\makeatother

\usepackage[linesnumbered,ruled,vlined]{algorithm2e}
\usepackage{algpseudocode}
\usepackage{amsmath}

\usepackage{fancyhdr}

\begin{document}

\title{Velocity-adaptive Access Scheme for MEC-assisted Platooning Networks: Access Fairness Via Data Freshness}

\author{Qiong Wu,~\IEEEmembership{Member,~IEEE}, Ziyang Wan, Qiang Fan, Pingyi Fan,~\IEEEmembership{Senior Member,~IEEE}  \\
and Jiangzhou Wang,~\IEEEmembership{Fellow,~IEEE}

\thanks{This work was supported in part by the National Natural Science Foundation of China under Grant No. 61701197, in part by the Beijing Natural Science Foundation under Grant No. 4202030, in part by the 111 Project under Grant No. B12018. (Corresponding authors: Qiong Wu)

Qiong Wu and Ziyang Wan are with the School of Internet of Things Engineering, Jiangnan University, Wuxi 214122, China (e-mail: qiongwu@jiangnan.edu.cn, ziyangwan@stu.jiangnan.edu.cn).

Qiang Fan is with Wistron AiEdge, San Jose, CA 95131, USA (e-mail: qiang\_fan@wistron.com).

Pingyi Fan is with the Department of Electronic Engineering, Beijing National Research Center for Information Science and Technology, Tsinghua University, Beijing 100084, China (email: fpy@tsinghua.edu.cn).

Jiangzhou Wang is with the School of Engineering and Digital Arts, University of Kent, CT2 7NT Canterbury, U.K. (e-mail: j.z.wang@kent.ac.uk).
}
\thanks{Manuscript received XXX, XX, 2015; revised XXX, XX, 2015.}}

\markboth{IEEE Internet of Things Journal,~Vol.~XX, No.~XX, XXX~2021}
{WU \MakeLowercase{\textit{et al.}}: Velocity-adaptive Access Scheme for MEC-assisted Platooning Networks: Access Fairness Via Data Freshness}

\maketitle

\begin{abstract}
Platooning strategy is an important part of autonomous driving technology. Due to the limited resource of autonomous vehicles in platoons, mobile edge computing (MEC) is usually used to assist vehicles in platoons to obtain useful information, increasing its safety. Specifically, vehicles usually adopt the IEEE 802.11 distributed coordination function (DCF) mechanism to transmit large amount of data to the base station (BS) through vehicle-to-infrastructure (V2I) communications, where the useful information can be extracted by the edge server connected to the BS and then sent back to the vehicles to make correct decisions in time. However, vehicles may be moving on different lanes with different velocities, which incurs the unfair access due to the characteristics of platoons, i.e., vehicles on different lanes transmit different amount of data to the BS when they pass through the coverage of the BS, which also results in the different amount of useful information received by various vehicles. Moreover, age of information (AoI) is an important performance metric to measure the  freshness of the data. Large average age of data implies not receiving the useful information in time. It is necessary to design an access scheme to jointly optimize the fairness and data freshness. In this paper, we formulate a joint optimization problem in the MEC-assisted V2I networks and present a multi-objective optimization scheme to solve the problem through adjusting the minimum contention window under the IEEE 802.11 DCF mode according to the velocities of vehicles. The effectiveness of the scheme has been demonstrated by simulation.
\end{abstract}


\IEEEpeerreviewmaketitle

\section{Introduction}
\label{sec1}
\IEEEPARstart{A}{utonomous} driving is a promising technology to change people's life, which is an important part of a smart city \cite{Plageras2018}. Some enterprises such as Google, Tesla and Baidu have engaged in developing advanced driverless technology. It is estimated that 25 percent of vehicles on the road will be driverless vehicles by 2035 \cite{BierstedtJ2018}.

Platooning strategy is an important management strategy for autonomous driving technology \cite{NieS2018,Zhao2021,Guo2021}. Autonomous vehicles adopt the platooning strategy to form platoons on a common lane. Each platoon is composed of a leader vehicle and some common vehicles. The leader vehicle controls the kinestate information of the entire platoon such as driving direction, velocity and acceleration. The member vehicles in the platoon follow the leader vehicle one after another with the same velocity and intra-platoon spacing \cite{JiaD2016}. In order to perceive the surrounding environment, the autonomous vehicles in the platoon are usually equipped with high-definition resolution cameras and LiDARs to collect environmental information such as the conditions of road, vehicles and pedestrians \cite{XuW2018}, and further generate high-resolution HD maps with high rate, e.g., {20-40Mbps} for high-definition cameras and {10-100Mbps} for LiDAR \cite{ChoiJ2016}. The large amount of data are usually redundant and needs to be analyzed to extract the useful information. However, the storage and computing capabilities of autonomous vehicles are usually insufficient to handle the data extraction \cite{XuW2018}. Cloud computing (CC) technology is an advanced data processing technology to solve such problem \cite{SpringerBook,1Stergiou2021}. It can be combined with big data (BD) to provide more continuous calculations \cite{Stergiou2018}. However, traditional cloud computing usually results in high latency \cite{Islam2013}, thus being unsuitable for the moving vehicle scenarios. Therefore, mobile edge computing (MEC) is employed as edge cloud (EC) to provision computing resources for vehicles with low delay \cite{2Stergiou2021}. Specifically, vehicles usually adopt the IEEE 802.11 distributed coordination function (DCF) mechanism to transmit such large amount of data to BS through vehicle-to-infrastructure (V2I) communications. As a BS is connected with an edge server that has rich sufficient computing resources, the useful information can be extracted by the edge server and then sent back to the autonomous vehicles \cite{FanQ2017}. In this way, autonomous vehicles can obtain useful information and make correct decisions in time.

Due to the unique characteristics of platoons, vehicles on the same lane has to move with the same velocity. The vehicle of high velocity will traverse the coverage area of a certain BS within a relatively short time duration as compared to other slow vehicles. Thus, the vehicle with high velocity may transmit less data to the BS than a vehicle with low velocity, incurring the unfair access under the ordinary IEEE 802.11 DCF. The unfair access will cause that the vehicles with high velocity receives less useful information than the vehicles with low velocity, and thus the probability that they cannot make correct decision in time is becoming relatively high, increasing the degree of unsafety. Therefore, it is critical to design a proper access scheme to ensure fair access. Meanwhile, age of information (AoI) is an important performance metric in vehicular network \cite{SanjitK2012_1,SanjitK2012_2}. Different from the traditional performance metrics such as latency and throughput \cite{QWu2016,HLiu2020,Jiang2019}, it represents the freshness of the data in the network. If the average age of data in the network is large, the data received by BS are averagely generated long time ago. Thus, the vehicle could not retrieve the latest useful information and make decision in time. To the best of our knowledge, no research has focused on the joint optimization of the access fairness and freshness of data in the MEC-assisted platooning network, which motivates us to conduct this work.

In this paper, we present a multi-objective optimization scheme for the MEC-assisted platooning network to ensure the fair access and relatively low average age of data through adjusting the minimum contention window according to the velocities of vehicles. The main contribution is summarized as follows.

\begin{itemize}
\item[1)] We formulate a joint optimization problem of access fairness and data freshness for V2I networks where the IEEE 802.11 DCF mode is employed and propose a new adaptive scheme, which can be implemented by adjusting the contention window according to the vehicle velocities.
\item[2)] We also define a fairness index to reflect the fairness of vehicles with different velocities. The paper shows that the fairness index can be expressed as a function of the velocity and minimum contention window.
\item[3)] The stochastic hybrid system (SHS) method is firstly adopted to model the transition of the age for a vehicle during the access process in this paper. The important finding is that it can exactly characterize the relationship between the average age of data in the network and minimum contention window.
\item[4)] For the associated multi-objective optimization problem, we adopt the multi-objective particle swarm optimization (MoPSO) algorithm to solve it and get the optimal minimum contention window for each vehicle.
\end{itemize}

The rest of this paper is organized as follows. Section \ref{sec2} reviews the related work. Section \ref{sec3} briefly describes the system model. Section \ref{sec4} introduces the fairness index and average age of data in the network. Section \ref{sec5} formulates the multi-objective optimization problem and adopts MoPSO algorithm to solve it and get the optimal contention window. We present some simulation results in Section \ref{sec6}, and conclude this paper in Section \ref{sec7}.

\section{Related Work}
\label{sec2}
In this section, we first review the existing works on the platooning network, various fairness of the network and age of information.

\subsection{Platooning network}
There have been many works to improve the performance of platooning network \cite{GeH2020,UcarS2018,XuL2019,GuoG2016,PengH2017,KazemiH2018,Campolp2017,GaoF2018}.
In \cite{GeH2020}, Wu \textit{et al.} considered the effect of the disturbance on the stability of platoons to construct an effective model to analyze the time-related performance of 802.11p-based platooning network.
In \cite{UcarS2018}, Ucar \textit{et al.} considered the security vulnerabilities of IEEE 802.11p and  visible light communication (VLC) to propose a hybrid security protocol based on IEEE 802.11p and VLC for platooning network to ensure the stability of the platoon.
In \cite{XuL2019}, Xu \textit{et al.} investigated the stability of platoons on the corners and provided an integrated platoon control framework for heterogeneous vehicles on curved roads.
In \cite{GuoG2016}, Guo \textit{et al.} considered the platoon control problem affected by capacity constraints and random packet loss to establish a framework for network access scheduling and proposed a platooning control scheme through adapting the binary sequences method.
In \cite{PengH2017}, Peng \textit{et al.} proposed a sub-channel allocation scheme and a power control mechanism for long term evolution (LTE) based vehicle-to-vehicle communication under multi-platoon scenario.
In \cite{KazemiH2018}, Kazemi \textit{et al.} proposed a neural network-based insertion detection and trajectory prediction scheme to prevent interfering vehicles from inserting into cooperative adaptive cruise control (CACC) platoon, or the leader vehicle from performing special operations such as braking.
In \cite{Campolp2017}, Campolp \textit{et al.} exploited pooled LTE resources to coordinate the platoon through the leader vehicle and proposed a solution that can meet the ultra-low latency requirements for message transmission in platoon.
In \cite{GaoF2018}, Gao \textit{et al.} proposed a distributed adaptive sliding mode control scheme to adapt to the changes of vehicle topology in the platooning control system.
However, the above works focused on the research of platooning network to solve the problem of driving control. They did not consider the effects of access fairness and age of useful information on the driving control.
\begin{figure*}[htbp]
\centering
\includegraphics[width=\linewidth, scale=1.00]{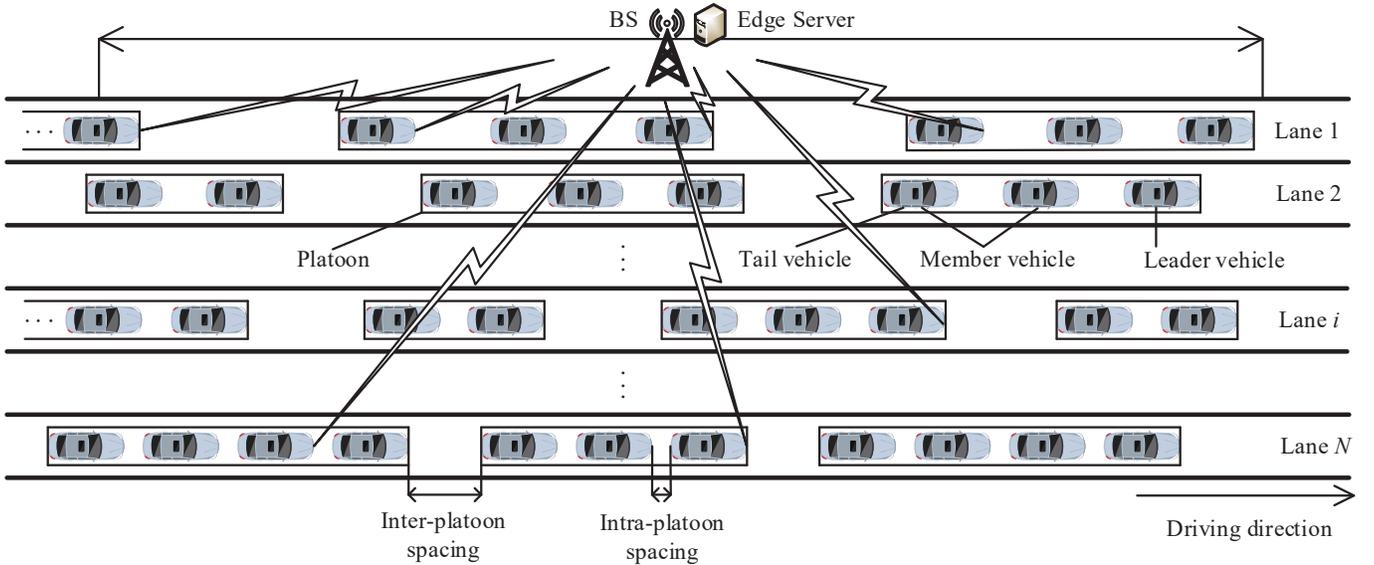}
\caption{System Model}
\label{fig1}
\vspace{-0.7cm}
\end{figure*}

\subsection{Fairness of network}
Some works have designed schemes to ensure various fairness of wireless network \cite{ChaJ2016,Ding2020,Rastegar2018,XiongK2017}.
In \cite{ChaJ2016}, Cha \textit{et al.} proposed a novel protocol to solve the unfair channel access caused by the increment of the contention window of the 802.11 protocol in the uplink wireless local area network (WLAN) with multiple packet reception (MPR).
In \cite{Ding2020}, Ding \textit{et al.} studied the fairness of the throughput generated by mobile ground users (GUs) when they communicate with unmanned aerial vehicles (UAVs) and the energy consumption of UAVs, and proposed an energy-efficient fair communication through trajectory design and band allocation (EEFC-TDBA) algorithm based on deep reinforcement learning (DRL) to autonomously choose the frequency band allocated to each mobile GU and flight direction of UAVs, thus the energy saving and fairness of throughput between GUs can be achieved.
In \cite{Rastegar2018}, Rastegar \textit{et al.} studied the fairness of users' allocation of ternary content addressable memory (TCAM) resources due to the flow table restriction in software-defined networks (SDN), and proposed an optimal two-step procedure to solve the unfair resource allocation problem.
In \cite{XiongK2017}, Xiong \textit{et al.} considered the near-far effect caused by unfair uneven power distribution in the high-speed railway communication (HSRC) scenario and proposed a new $\beta$-fair power distribution scheme to achieve the fairness of power distribution while ensuring high mobile traffic.

A few works have proposed schemes to achieve the fair V2I communication caused by different velocities of vehicles \cite{Karamad2008,Harigovindan2013,WuQ2015,XiaSY2018}.
In \cite{Karamad2008}, Karamad \textit{et al.} proposed an improved 802.11 DCF mechanism to ensure the fair V2I communication in VANETs.
In \cite{Harigovindan2013}, Harigovindan \textit{et al.} used proportional fairness (PF) resource allocation to solve the unfair V2I communication when vehicles with different velocities transmit with different transmission rates and found the optimal minimum contention window to ensure that the channel resource can be distributed fairly in single-channel and multi-channel networks.
In \cite{WuQ2015}, the authors considered the non-saturated condition of the network, i.e., each vehicle does not always have packet to transmit, and constructed a model for 802.11 DCF based communication in VANETs to ensure the fair V2I communication.
In \cite{XiaSY2018}, the authors considered the non-saturated condition of the network and derived a model to achieve the fair V2I communication in platooning scenario.
However, these works did not consider the effects of age of information.

\subsection{Age of Information}
Age of information has been paid wide attentions in recent years as a new communication metric to measure the data freshness \cite{SanjitK2011,NiYuan2018,Abd-Elmagid2019,Samir2020,WangM2019,AhmedM2017,RoyD2019,AliMaatouk2020}.
In \cite{SanjitK2011}, Sanjit \textit{et al.} first considered the age of information in the vehicle network architecture. They proposed a broadcast rate adaptive algorithm at the application layer to reduce the age of information.
In \cite{NiYuan2018}, Ni \textit{et al.} studied the age of information in the network for the beacon broadcast schedule in VANET and proposed an algorithm to minimize the age of information considering the limited communication resources and the mobility of vehicles.
In \cite{Abd-Elmagid2019}, Abd-Elmagid \textit{et al.} jointly optimized the UAVs flight trajectory, distribution of energy and service time for data packet transmission to minimize the average peak of the age in the UAV-assisted network.
In \cite{Samir2020}, Smair \textit{et al.} designed the trajectory and scheduling strategy problem as a markov decision process (MDP) for the UAV-assisted networks. They used the deep deterministic policy gradient (DDPG) algorithm to learn the trajectory of the deployed UAV and obtained the optimal scheduling strategy to effectively minimize the expected weighted sum AoI (EWSA).
In \cite{WangM2019}, Wang \textit{et al.} considered the wireless sensor network where all nodes use the carrier sense multiple access/collision avoidance (CSMA/CA) mechanism to send their updates in the network. They derived the transmission model of the station and adopted the scaling law to adjust the frame length according to the node density to minimize the average broadcast age of information (BAoI) in the wireless sensor network.
In \cite{AhmedM2017},  Ahmed \textit{et al.} proved that there is a preemptive last generated first served (LGFS) strategy that can ensure the minimum age of the multi-hop WLAN.
In \cite{RoyD2019}, Roy \textit{et al.} considered a simple queue status updates from multiple source stations and calculated the age of information for each source station through SHS.
In \cite{AliMaatouk2020}, Ali \textit{et al.} analyzed the service rate and back-off rate in the WLAN when the CSMA/CA mechanism is used and obtained the closed-form solution of the age of information for each link, then transformed the original optimization problem into a convex problem, and finally obtained the optimal minimum contention window to make the minimum sum age of information of the network.
However, these works did not consider the fairness of the network.

As mentioned above, there is no work designing a scheme to jointly optimize the access fairness and freshness of data for platooning network, which is crucial for autonomous driving.

\section{System Model}
\label{sec3}
In this section, we will describe the system model in detail, which is shown in Fig. \ref{fig1}. Consider a  one way scenario of highway, which is composed of $N$ lanes. Vehicles on each lane are moving with the same direction. The BSs are deployed along the roadside of highway, where each BS is equipped with an edge server having powerful computing capabilities. Multiple platoons are straightly traversing the coverage area of the considered BS on their specific lanes. Each platoon consists of a leader vehicle and several member vehicles, where the leader vehicle controls the kinematics parameters of the platoon while the member vehicles follow it one by one. The tail vehicle is the last vehicle of the platoon. For simplicity, it is assumed that the vehicles on the same lane are homogeneous, i.e., the velocities, intra-platoon spacing and inter-platoon spacing are the same for the vehicles on the same lane and different for the vehicles on different lanes, where the intra-platoon spacing is the distance between two adjacent vehicles in a platoon, while the inter-platoon spacing denotes the distance between the tail vehicle of the previous platoon and the leader one of the following neighbor platoon. The platoons of each lane arrive at the coverage of the BS according to the Poisson process. Each vehicle transmits data to the BS, where the useful information can be extracted by the edge server. Then the useful information is sent back to the vehicles. Since the size of useful information is very small and the downlink channel capacity (BS to Vehicles) is relatively large as compared with the uplink capacity (Vehicles to BS), the downlink is not considered here. Similar to \cite{Marlon2007}, we assume the vehicles on the same lane have the same communication parameters.

Assume that the signal transceiver is installed on the headstock of each vehicle. Once a vehicle arrives at the network, i.e., the headstock of the vehicle arrives at the coverage of BS, it will transmit packet to the BS immediately. We consider the saturated condition of the network \cite{Bianchi2000}, i.e., each vehicle always has one packet to transmit when it captures the channel, to explore the extreme performance of our scheme. Each vehicle adopts the 802.11 DCF for data transmission. Specifically, a vehicle initializes a back-off process and randomly chooses an integer value from 0 to $W_0-1$ as the back-off counter, where $W_0$ is the minimum contention window. Then the value of the back-off counter decreases by 1 after each time slot. When the value of the back-off counter decreases to 0, the vehicle captures the channel. Then the vehicle generates a packet and transmits it to the BS via the captured channel. If more than one vehicle are transmitting at the same time, a collision will occur and thus incurs the transmission failure; otherwise the transmission is successful. The above process will be repeated to transmit different packets.

\section{Fairness and Age of Data}
\label{sec4}
In this section, we derive the relationship between the fairness, average age of data in the network and minimum contention windows as well as velocities. We first define the fairness index to measure the fairness of the network and derive the fairness index as a function of velocity and minimum contention window. Then we adopt the SHS to derive the relationship between the average age of data in the network and minimum contention window as well as velocity. The parameters used in this section are listed in Table \ref{tab1}.

\begin{table}\footnotesize
\caption{Notations used in this section}
\label{tab1}
\centering
\begin{tabular}{|c|p{6.6cm}|}
\hline
\textbf{Notation} &\textbf{Description}\\
\hline
$\boldsymbol{A}_{l}$ & The transition reset matrix of age process. \\
\hline
$C_{bit}$ & The bit rate of the channel. \\
\hline
$C(\mathbf{R})$ & The normalization factor. \\
\hline
$C_s^i$ & The transmission rate of a vehicle on the $i$-th lane. \\
\hline
$d_w^i$ &The average length of a complete platoon on the $i$-th lane. \\
\hline
$D_s^i$ &\multicolumn{1}{m{6.6cm}|}{The average length of the incomplete platoon-interval on the $i$-th lane.} \\
\hline
$D_w^i$ &The average length of a platoon-interval on the $i$ lane. \\
\hline
$H_k$ &The average service rate of link $k$. \\
\hline
$K_{index}$ &The fairness index. \\
\hline
$\mathbb{L}_q$ &The set of all paths from state $q$. \\
\hline
$\mathbb{L}_q'$ &The set of all paths that transition to state $q$. \\
\hline
$m_i$ &The number of complete platoons on the $i$-th lane. \\
\hline
$n_i$ &The number of vehicles on the $i$-th lane. \\
\hline
$n_p$ &The average number of vehicles in a complete platoon. \\
\hline
$n_s^i$ &\multicolumn{1}{m{6.6cm}|}{The number of vehicles in the incomplete platoon on the $i$-th lane.} \\
\hline
$n_w^i$ &\multicolumn{1}{m{6.6cm}|}{The number of vehicles in complete platoon on the $i$-th lane.} \\
\hline
$N$ &The number of lanes within the coverage of the BS. \\
\hline
$N_{bit}$ &The average number of bits in a packet. \\
\hline
$N_v$ &The total number of vehicles in network. \\
\hline
$p_i$ &The collision probability on the $i$-th lane. \\
\hline
$q(t)$ &The discrete process of Markov chain. \\
\hline
$r_0$ &The minimum intra-platoon spacing. \\
\hline
$r_s^i$ &The average inter-platoon spacing on the $i$-th lane. \\
\hline
$r_w^i$ &The intra-platoon spacing on the $i$-th lane. \\
\hline
$R$ &The BS coverage. \\
\hline
$R_k$ &The average back-off rate of link $k$. \\
\hline
$s$ &The average length of a vehicle. \\
\hline
$S$ &The normalized throughput of network. \\
\hline
$T_b^i$ &The average back-off time of a vehicles on the $i$-th lane. \\
\hline
$T_h$ &The minimum time headway. \\
\hline
$T_i$ &\multicolumn{1}{m{6.6cm}|}{The time that a vehicle on the $i$-th lane takes to pass through the coverage of BS in the network.} \\
\hline
$T_s$ &The average time of successful transmission. \\
\hline
$T_{slot}$ &The length of a time slot. \\
\hline
$v_i$ &The velocity of a vehicle on the $i$-th lane. \\
\hline
$\overline{\boldsymbol{v}}_{ql}$ &The steady correlation between $q_l$ and $\boldsymbol{x}$. \\
\hline
$W_0^i$ &\multicolumn{1}{m{6.6cm}|}{The minimum contention window of a vehicle on the $i$-th lane.} \\
\hline
$\boldsymbol{x}(t)$ &The age for a link at $t$. \\
\hline
$\Delta T_i$ &\multicolumn{1}{m{6.6cm}|}{The time difference between two consecutive platoons on the $i$-th lane to enter the network.} \\
\hline
$\overline{\Delta}$ &The average age of data in the network. \\
\hline
$\overline{\Delta}_k$ &The average age of data for link $k$. \\
\hline
$\lambda_i$ &The platoon arrival rate on the $i$-th lane. \\
\hline
$\lambda_i^{max}$ &The maximum platoon arrival rate on lane $i$. \\
\hline
$\lambda^{(l)}$ &The transition rate from state $q$ to state $q_l$. \\
\hline
$\tau_i$ &The transmission probability of vehicles on the $i$-th lane. \\
\hline
$\overline{\pi}_q$ &\multicolumn{1}{m{6.6cm}|}{The state probability of the Markov chain of the discrete process $q(t)$.} \\
\hline
\end{tabular}
\vspace{-0.6cm}
\end{table}

\subsection{Fairness Index}
The fairness of the network means that the amount of data transmitted by vehicles with different velocities are the same when they are traversing the coverage of BS, thus we have
\begin{equation}
{C_s^i}{T_i} = C,
\label{eq1}
\end{equation}
where $T_i$ is the time duration one vehicle on lane $i$ ${(i = 1, 2, \ldots N)}$ traverses the coverage of a BS, $C_s^i$ is the transmission rate of a vehicle on lane $i$ $(C$ is a constant$)$. Since vehicles on the same lane keep with the same velocity, $T_i$ can be calculated as

\begin{equation}
T_i = \frac{R}{v_i},
\label{eq2}
\end{equation}
where $v_i$ is the velocity of a vehicle on lane $i$ and $R$ is coverage range of BS.

According to \cite{WuQ2015}, the transmission rate of a vehicle on lane $i$ is calculated as
\begin{equation}
{C_s^i} = S \times \frac{{C_{bit}}}{N_{bit}} \times \frac{\tau_i}{{\sum\limits_{j=1}^{N_v}{\tau_j}}},
\label{eq3}
\end{equation}
where $S$ is the normalized throughput of the network, $C_{bit}$ is the bit rate of the channel, $N_{bit}$ is the bits of each data packet, $N_v$ is the total number of vehicles in the network and $\tau_i$ is the transmission probability of a vehicle on lane $i$.

Substituting Eqs. \eqref{eq2} and \eqref{eq3} into Eq. \eqref{eq1}, we have
\begin{equation}
S \times \frac{{C_{bit}}}{N_{bit}} \times \frac{\tau_i}{{\sum\limits_{j=1}^{N_v}{\tau_j}}} \times \frac{R}{v_i} = C,
\label{eq4}
\end{equation}
where $S$, $C_{bit}$, $N_{bit}$, $\sum\limits_{j=1}^{N_v}{\tau_j}$ and $C$ are constants. Thus Eq. \eqref{eq4} can be rewritten as
\begin{equation}
\tau_i \times \frac{R}{v_i} = \frac{C}{S \times \frac{C_{bit}}{N_bit} \times \frac{1}{\sum\limits_{j=1}^{N_v}{\tau_j}}} = \frac{C}{C'} = K_{index}^i.
\label{eq5}
\end{equation}
where
\begin{equation}
C' = S \times \frac{{C_{bit}}}{N_{bit}} \times \frac{1}{{\sum\limits_{j=1}^{N_v}{\tau_j}}},
\label{eq6}
\end{equation}
which is independent of lane $i$ and is a constant. Since $C$ is a constant, $K_{index}^i$ in Eq. \eqref{eq5} is also a constant. In other words, $K_{index}^i$ should be a constant for the vehicles on different lanes to achieve fairness. $K_{index}^i$ is denoted as the fairness index on lane $i$.

According to \cite{Bianchi2000}, the relationship between $\tau_i$ and the minimum contention window of a vehicle on lane $i$, i.e., $W_0^i$, in the saturated network is calculated as
\begin{equation}
\begin{aligned}
\tau_i = \frac{2}{W_0^i +1}.
\label{eq7}
\end{aligned}
\end{equation}

By substituting Eq. \eqref{eq7} into Eq. \eqref{eq5}, we can get the relationship between $K_{index}^i$ and $v_i$, $W_0^i$, i.e.,
\begin{equation}
K_{index}^i = \frac{2R}{v_i (W_0^i + 1)}.
\label{eq8}
\end{equation}

Given velocities of vehicles on each lane, we can adaptively adjust the minimum contention window $W_0^i$ to design the fairness index at lane $i$, i.e., $K_{index}^i$, under the given velocity $v_i$.

Averaging both sides of Eq. \eqref{eq8}, we have
\begin{equation}
K_{index} = \frac{2R}{\overline{v} (\overline{W}_0 + 1)},
\label{eq9}
\end{equation}
where $\overline{v}$ is the average velocity of vehicles in the network which can be calculated under given the velocities of vehicles on each lane, $\overline{W}_0$ is the average minimum contention window of the network which is given initially in the system, thus $K_{index}$ can be obtained according to Eq. \eqref{eq9}. $K_{index}$ is used to measure the fairness of the network, which is referred to as fairness index of the network. In the network, the access approaches to fairness when the fairness index at each lane $i$, i.e., $K_{index}^i$, gets close to the fairness of the network $K_{index}$.

We can obtain the condition to achieve the fairness in the network, which is described as Lemma \ref{lemma1}, i.e.,
\newtheorem{lemma}{Lemma}[]
\begin{lemma}
\label{lemma1}
The fairness of the network is achieved when ${v_1W_0^1 \approx v_2W_0^2 \approx \dots \approx v_NW_0^N \approx \overline{v}\overline{W}_0}$.
\end{lemma}
\begin{proof}
In Eqs. \eqref{eq8} and \eqref{eq9}, $v_iW_0^i$ and $\overline{v}\overline{W}_0$ are much larger than 1, thus $K_{index}^i \approx \frac{2R}{v_iW_0^i}$ and $K_{index} \approx \frac{2R}{\overline{v}\overline{W}_0}$. Since the fairness of the network is achieved when the fairness index at each lane $i$ equals to the fairness of the network, i.e., $K_{index}^i = K_{index}$ is satisfied for each lane $i$, thus based on the approximated $K_{index}^i$ and $K_{index}$, the fairness of the network is achieved when ${v_1W_0^1 \approx v_2W_0^2 \approx \dots \approx v_NW_0^N \approx \overline{v}\overline{W}_0}$.
\end{proof}

\begin{figure}[htbp]
\centering
\includegraphics[width=\linewidth, scale=1.00]{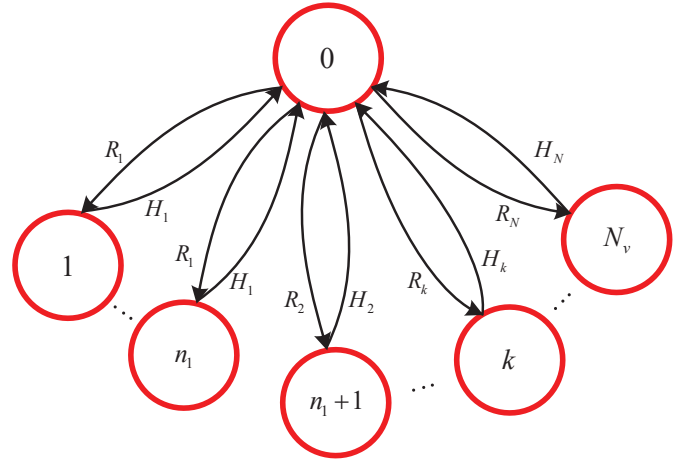}
\caption{Markov chain of the discrete process $q(t)$}
\label{fig2}
\vspace{-0.6cm}
\end{figure}
\subsection{Average Age of Data in the Network}
In this sub-section, we further derive the average age of data in the network, which is equal to the average age of data for each transmitter-receiver pair, i.e., a vehicle and the BS. A transmitter-receiver pair is referred to as a transmission link in this paper. Here, link $k$, represents the communications link between vehicle $k$ and BS, and we adopt the SHS approach to model its transmission process \cite{RoyD2019,AliMaatouk2020}. Similar to \cite{SunY2017}, the data sampling time could be neglected in the modeling since the time interval required for sampling is usually much less than for data transmission.

We define the state set $(q(t), \boldsymbol{x}(t))$ to model the system. Specifically, $q(t)\in \{ 0, 1, 2, \dots N_v \}$ is used to represent the system state at time slot $t$, where $N_v$ indicates the total number of vehicles in the network, where ${q(t) = 0}$ indicates that no link captures the channel at $t$, i.e., the channel is idle at $t$, and ${q(t) = k} (k\neq 0)$ indicates that link $k$ captures the channel at $t$; $\boldsymbol{x}(t)= [x_0(t), x_1(t)]$ is used to represent the age of data for link $k$ at $t$, where $x_0(t)$ is the age of data at the receiver, i.e., BS, at $t$ and $x_1(t)$ is the age of data at link $k$ at $t$. The age of link $k$ at BS is initialized at the beginning and increases linearly with slope of one. The age of data at link $k$ is initialized when vehicle $k$ generates the packet, and then increases linearly with slope of one. When the BS receives the packet, the age of data at receiver is reset to the age of data at link $k$. According to the definition above, the integer variable $q(t)$ maps to a discrete process while $\boldsymbol{x}(t)$ maps to a continuous process.

Let $H_k$ be the average service rate, i.e., the average rate of successful transmission of link $k$, and $R_k$ be the average back-off rate of link $k$. The channel would be idle after link $k$ transmits a packet successfully, and thus state $k$ transits to $0$ with rate $H_k$. Vehicle $k$ would capture the channel to transit a packet after the back-off process and thus state $0$ transits to $k$ with rate $R_k$. Let the number of vehicles on lane $i$ be $n_i$. As described in the Section \ref{sec3}, we consider a scenario composed of $N$ lanes, where vehicles on the same lane are homogeneous. In this case, state $0$ transits to anyone state from $1$ to $n_1$ with the same rate $R_1$, and every state from $1$ to $n_1$ transits to state $0$ with the same rate $H_1$. The transitions between state $0$ and each state from $n_1+1$ to $N_v$ can be deduced similarly. However, Fig. \ref{fig2} cannot reflect the transition of $\boldsymbol{x}(t)$. Next, we will further describe the transition of $\boldsymbol{x}(t)$.

According to SHS, a transition of the discrete process would incur a reset of the continuous process. Let $l$ be a transition of the discrete process, for which we have ${\boldsymbol{x}' = \boldsymbol{x} \boldsymbol{A}_{l}}$, where $\boldsymbol{x}$ and $\boldsymbol{x}'$ are the continuous processes before and after resetting, respectively, and $\boldsymbol{A}_{l}$ is the transition reset maps. Let $q_l$ and $q_l'$ be the discrete state before and after transition $l$, respectively. Then, ${q_l \to q_l'}$ represents the state transition between state $q_l$ and $q_l'$, $\lambda^{(l)}$ is the transition rate of transition $l$. Moreover, since the Markov process in the SHS is ergodic \cite{RoyD2019}, we have  ${\overline{\boldsymbol{v}}_{ql}=[\overline{v}_{q0},\overline{v}_{q1}]}$, where $\overline{\boldsymbol{v}}_{ql}$ is the steady correlation between $q_l$ and $\boldsymbol{x}=[x_0,x_1]$ for transition $l$, here $\overline{v}_{q0}$ is the steady correlation between $q_l$ and $x_0$, $\overline{v}_{q1}$ is the steady correlation between $q_l$ and $x_1$. According to the reset of the continuous process, we can further derive ${\overline{\boldsymbol{v}}_{ql}' = \overline{\boldsymbol{v}}_{ql} \boldsymbol{A}_{l}}$, where $\overline{\boldsymbol{v}}_{ql}'$ is the steady correlation between $q_l'$ and $\boldsymbol{x}'$. Based on the above, we can employ Table \ref{tab2} to describe the SHS.

\begin{table}[htbp]\footnotesize
\centering
\caption{SHS description}
\begin{tabular}{cccccc}
$l$ & $q_l \to q_l'$ & $\lambda^{(l)}$ & $\boldsymbol{x}'=\boldsymbol{x}\boldsymbol{A}_l$ & $\boldsymbol{A}_l$ & $\overline{\boldsymbol{v}}_{ql}' = \overline{\boldsymbol{v}}_{q_l} \boldsymbol{A}_l$ \\ \hline
$1$ & $0 \to 1$ & $R_1$ & ${[x_0,x_1]}$ & $\Big [ \begin{matrix} 1 & 0 \\ 0 & 1 \end{matrix} \Big ]$ & ${[\overline{v}_{00},\overline{v}_{01}]}$ \\
$\vdots$ & $\vdots$ & $\vdots$ & $\vdots$ & $\vdots$ & $\vdots$ \\
$N_v$ & $0 \to N_v$ & $R_{N_v}$ & ${[x_0,x_1]}$ & $\Big [ \begin{matrix} 1 & 0 \\ 0 & 1 \end{matrix} \Big ]$ & ${[\overline{v}_{00},\overline{v}_{01}]}$ \\
$N_v + 1$ & $1 \to 0$ & $H_{1}$ & ${[x_0,x_1]}$ & $\Big [ \begin{matrix} 1 & 0 \\ 0 & 1 \end{matrix} \Big ]$ & ${[\overline{v}_{10},\overline{v}_{11}]}$ \\
$\vdots$ & $\vdots$ & $\vdots$ & $\vdots$ & $\vdots$ & $\vdots$ \\
$N_v + k$ & $k \to 0$ & $H_{k}$ & ${[x_1,0]}$ & $\Big [ \begin{matrix} 0 & 0 \\ 1 & 0 \end{matrix} \Big ]$ & ${[\overline{v}_{k1},0]}$ \\
$\vdots$ & $\vdots$ & $\vdots$ & $\vdots$ & $\vdots$ & $\vdots$ \\
$2N_v$ & $N_v \to 0$ & $H_{N_v}$ & ${[x_0,x_1]}$ & $\Big [ \begin{matrix} 1 & 0 \\ 0 & 1 \end{matrix} \Big ]$ & ${[\overline{v}_{N_v0},\overline{v}_{N_v1}]}$ \\ \hline
\end{tabular}
\label{tab2}
\vspace{-0.2cm}
\end{table}

Next we explain the transitions in Table \ref{tab2} in detail.
\begin{itemize}
\item[1)] Transition $l_a$ $(l_a = \{0, 1, 2, \dots, k, \dots, N_v \})$ means that the channel has been captured. In this case, the transition rate for transition $l_a$ is $R_{l_a}$. Since link $k$ will not reset its age after any link captures its channel, we have ${\boldsymbol{x}' =\boldsymbol{x}\boldsymbol{A}_{l_a}= \boldsymbol{x}}$ and $\overline{\boldsymbol{v}}_{ql_a}' = \overline{\boldsymbol{v}}_{q_{l_a}} \boldsymbol{A}_{l_a}=\overline{\boldsymbol{v}}_{q_{l_a}}$.
\item[2)] Transition $l_b$ $(l_b = \{N_v + 1, N_v + 2, \dots, N_v + k, \dots 2N_v\})$ means that the channel becomes idle after the successful transmission. In this case, the transition rate for transition $l_b$ is $H_{l_b-N_v}$. For transition ${l_b = N_v + k}$, the age of data at the receiver is reset to $x_1$ and the age of data at link $k$ is reset to 0 because of the successful transmission, and thus ${\boldsymbol{x}' =\boldsymbol{x}\boldsymbol{A}_{l_b}= [x_1,0]}$ and $\overline{\boldsymbol{v}}_{q{l_b}}' = \overline{\boldsymbol{v}}_{q_{l_b}} \boldsymbol{A}_{l_b}=[\overline{v}_{k1},0]$. For other transitions, the successful transmission of any link will not reset the age for link $k$, i.e., ${\boldsymbol{x}' =\boldsymbol{x}\boldsymbol{A}_{l_b}= \boldsymbol{x}}$ and $\overline{\boldsymbol{v}}_{q{l_b}}' = \overline{\boldsymbol{v}}_{q_{l_b}} \boldsymbol{A}_{l_b}=\overline{\boldsymbol{v}}_{q_{l_b}}$.
\end{itemize}

Based on the analysis of \cite{RoyD2019}, the average age of link $k$ is calculated as
\begin{equation}
\begin{aligned}
\overline{\Delta}_k = \sum_{q} \overline{v}_{q0},\quad & \forall k \in \{1, \ldots N_v \},\\[-.6pc]
& \forall q \in \{0, \ldots N_v\}.
\label{eq10}
\end{aligned}
\end{equation}

According to Eq. \eqref{eq10}, $\overline{v}_{q0}$ needs to be further analyzed to achieve the average age of a link.

At first, based on the theorem given in \cite{RoyD2019}, we have
\begin{equation}
\begin{aligned}
\overline{\boldsymbol{v}}_{ql_1} (\sum_{l_1} \lambda^{(l_1)}) = \boldsymbol{b}_q \overline{\pi}_q + \sum_{l_2} \lambda^{(l_2)} \overline{\boldsymbol{v}}_{ql_2} \boldsymbol{A}_{l_2}, \quad &l_1 \in \mathbb{L}_q, \\[-.6pc] & l_2 \in \mathbb{L}_q'.
\label{eq11}
\end{aligned}
\end{equation}
where $\overline{\pi}_q$ is the steady state probability of state $q$, $\mathbb{L}_q'$ and $\mathbb{L}_q$ are the input and output set of each discrete state, respectively, $\boldsymbol{b}_q$ is a binary differential equation vector of the age evolution at state $q$, i.e., ${\boldsymbol{b}_q}=\dot{\boldsymbol{x}}$.

Since vehicle $k$ generates a packet to transmit only when it captures the channel, the age of data at link $k$ is increased with slope of one when $q=k$ and it keeps $0$ when $q \neq k$. Moreover, the age at the receiver always increases linearly with slope of one. Therefore, we have
\begin{equation}
\left\{ \begin{array}{lcl}
\boldsymbol{b}_q = [1,0], \quad \forall q \neq k, \\
\boldsymbol{b}_q = [1,1], \quad q = k.
\end{array} \right .
\label{eq12}
\end{equation}

Next, we will further derive $\overline{v}_{q0}$ according to Eqs. \eqref{eq11}, \eqref{eq12} and Table \ref{tab2}.
When $q=0$, we have $\boldsymbol{b}_q = [1,0]$ from Eq. \eqref{eq12}. In this case, the left hand side of Eq. \eqref{eq11} reflects the transitions from state $0$ to other states, i.e., from transition $1$ to $N_v$ in Table \ref{tab2}. The right hand side of Eq. \eqref{eq11} reflects the transitions from other states to state $0$, i.e., from transition $N_v+1$ to $2N_v$ in Table \ref{tab2}. Thus according to Table \ref{tab2}, we have
\begin{equation}
\overline{v}_{00} (\sum_{j=1}^{N_v} R_j) = \overline{\pi}_0 + \sum_{\substack{j=1 \\ j \neq k}}^{N_v} H_j \cdot \overline{v}_{j0} + H_k \cdot \overline{v}_{k1},
\label{eq13}
\end{equation}
\begin{equation}
\overline{v}_{01} (\sum_{j=1}^{N_v} R_j) = \sum_{\substack{j=1 \\ j \neq k}}^{N_v} H_j \cdot \overline{v}_{j1}.
\label{eq14}
\end{equation}

Similarly, when $q \neq 0$, the left hand side of Eq. \eqref{eq11} reflects the transition from state $q$ to state $0$, i.e., transition $N_v+q$ in Table \ref{tab2}. The right hand side of Eq. \eqref{eq11} reflects the transition from state $0$ to state $q$, i.e., transition $q$ in Table \ref{tab2}. Thus according to Table \ref{tab2} and Eq. \eqref{eq12}, we have
\begin{equation}
\overline{v}_{q0} \cdot H_q = \overline{\pi}_q + R_q \cdot \overline{v}_{00}, \quad \forall q = \{1, \ldots N_v\},
\label{eq15}
\end{equation}
\begin{equation}
\overline{v}_{q1} \cdot H_q = R_q \cdot \overline{v}_{01}, \quad \qquad q \neq k,
\label{eq16}
\end{equation}
\begin{equation}
\overline{v}_{k1} \cdot H_k = \overline{\pi}_k + R_k \cdot \overline{v}_{01},  \quad q = k.
\label{eq17}
\end{equation}

Next, we will derive $\overline{v}_{00}$ and $\overline{v}_{q0}$ according to Eqs. \eqref{eq13}-\eqref{eq17} to determine the average age of data in the network.

The relationship between $\overline{v}_{q0}$ and $\overline{v}_{00}$ can be obtained according to Eq. \eqref{eq15} , i.e.,
\begin{equation}
\overline{v}_{q0}=\frac{\overline{\pi}_q}{H_q}+\frac{R_q}{H_q}\cdot \overline{v}_{00}.
\label{eq18}
\end{equation}

Next, $\overline{v}_{00}$ is further derived to determine $\overline{v}_{q0}$. According to Eq. \eqref{eq16}, ${\overline{v}_{q1}}$ can be calculated as

\begin{equation}
{\overline{v}_{q1} = \frac{R_q}{H_q} \cdot \overline{v}_{01}}.
\label{eq19}
\end{equation}

\begin{figure*}[htbp]
\centering
\includegraphics[width=\linewidth, scale=1.00]{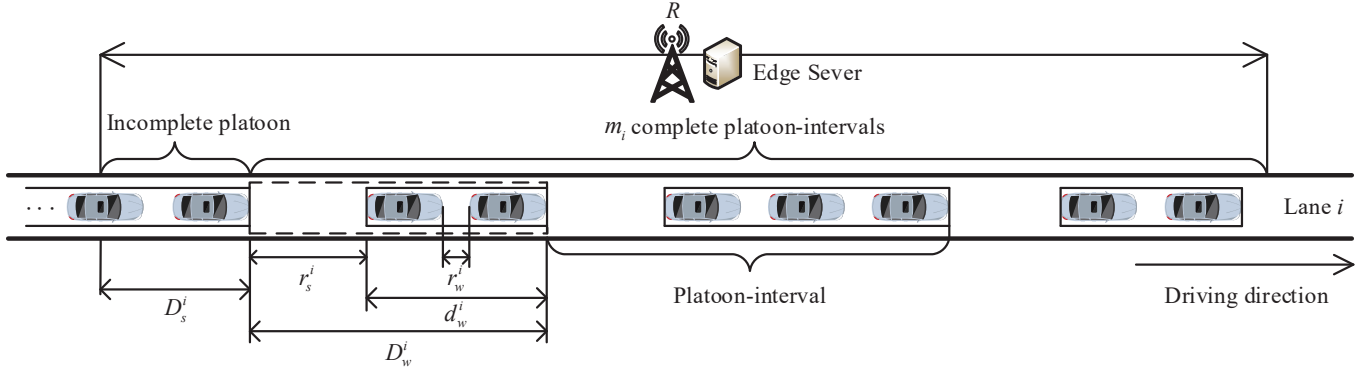}
\caption{Complete and incomplete platoon}
\label{fig3}
\vspace{-0.5cm}
\end{figure*}

Substituting Eq. \eqref{eq19} into Eq. \eqref{eq14}, we can get $\overline{v}_{01}=0$. Substituting $\overline{v}_{01}$ into Eq. \eqref{eq17}, we have
\begin{equation}
\overline{v}_{k1} = \frac{\overline{\pi}_k}{H_k}.
\label{eq20}
\end{equation}

Substituting Eqs. \eqref{eq18} and \eqref{eq20} into Eq. \eqref{eq13}, we have
\begin{equation}
\begin{aligned}
\overline{v}_{00} (\sum_{j=1}^{N_v} R_j) &= \overline{\pi}_0 + \sum_{\substack{j=1 \\ j \neq k}}^{N_v} (\overline{\pi}_j + R_j \cdot \overline{v}_{00}) + \overline{\pi}_k.
\label{eq21}
\end{aligned}
\end{equation}

Reorganizing Eq. \eqref{eq21}, $\overline{v}_{00}$ can be calculated as
\begin{equation}
\overline{v}_{00} = \frac{\overline{\pi}_0 + \sum_{j=1}^{N_v} \overline{\pi}_j}{R_k} = \frac{1}{R_k}.
\label{eq22}
\end{equation}

Combining Eqs. \eqref{eq18} and \eqref{eq22}, $\overline{v}_{q0}$ can be expressed as a function of $\overline{\pi}_q$,$R_q$ and $H_q$. Substituting Eqs. \eqref{eq22} and \eqref{eq18} into Eq. \eqref{eq10}, the average age for link $k$ can be calculated as
\begin{equation}
\begin{aligned}
\overline{\Delta}_k &= \overline{v}_{00} + \sum_{q=1}^{N_v} \overline{v}_{q0} = \sum_{q=1}^{N_v} \frac{\overline{\pi}_q}{H_q} + \sum_{q=1}^{N_v}(\frac{R_q}{H_q} +1 )\frac{1}{R_k}.
\label{eq23}
\end{aligned}
\end{equation}

According to \cite{AliMaatouk2020}, the steady state probability of state $q$ is calculated as
\begin{equation}
\left\{ \begin{array}{lcl}
\overline {\pi }_{0}=\displaystyle \frac {1}{C(\boldsymbol {R})}, \quad \\[.6pc]
\overline {\pi }_{k}=\displaystyle \frac {R_{k}}{C(\boldsymbol {R}) \cdot H_{k}}, \quad k \in [1, \ldots, N_v],
\end{array} \right .
\label{eq24}
\end{equation}
where $C(\boldsymbol{R})$ is a normalization factor that is expressed as
\begin{equation}
C(\boldsymbol{R}) = 1 + \sum_{k = 1}^{N_v} \frac{R_k}{H_k}.
\label{eq25}
\end{equation}

Substituting Eqs. \eqref{eq24} and \eqref{eq25} into Eq. \eqref{eq23}, the average age of data for link $k$ can be expressed as
\begin{equation}
\begin{aligned}
\overline{\Delta}_k = \frac{C(\boldsymbol{R}) }{R_k} + \sum_{q=1}^{N_v} \frac{R_q}{C(\boldsymbol{R}) \cdot {H_q}^2}.
\label{eq26}
\end{aligned}
\end{equation}

Since vehicles on the same lane have the same back-off rate and service rate, thus we can re-write the average age for any vehicle on lane $i$ as
\begin{equation}
\overline{\Delta}_i = \frac{C(\boldsymbol{R}) }{R_i} + \sum_{j=1}^{N} \frac{n_j \cdot R_j}{C(\boldsymbol{R})\cdot {H_j}^2}.
\label{eq27}
\end{equation}
where $n_j$ is the total number of vehicles on lane $j$. After averaging the ages for all links \cite{AliMaatouk2020}, we can get the average age of data in the network, i.e.,
\begin{equation}
\overline{\Delta} = \frac{1}{N_v} (\sum_{j=1}^{N_v} \frac{C(\boldsymbol{R}) }{R_j} + \frac{N_v }{C(\boldsymbol{R})} \sum_{j=1}^{N} n_j \frac{R_j}{{H_j}^2}).
\label{eq28}
\end{equation}

In addition, we give the condition to achieve the minimum age of information of the network in Lemma \ref{lemma2}, i.e.,
\begin{lemma}
\label{lemma2}
When the average service rate of the network is given, the minimum average age of information of the network is achieved when ${W_0^1 \approx W_0^2 \approx \dots \approx W_0^{N_v} \approx \overline{W}_0}$.
\end{lemma}
The proof of Lemma \ref{lemma2} please refer to APPENDIX at the end of this paper.

According to Eq. \eqref{eq28}, we can find the average age of data in the network is determined by the average back-off rate and service rate of a vehicle on lane $i$, i.e., $R_i$ and $H_i$, and the average number of vehicles in the network $N_v$. Next, we will further derive the relationship between $R_i$, $H_i$, $N_v$ and the minimum contention window $W_0^i$ as well as velocity $v_i$ of a vehicle on lane $i$.

\subsubsection{Average Back-off Rate and Service Rate}
\
\newline
\indent
We first derive the average back-off rate and service rate. The back-off rate on lane $i$ can be expressed as
\begin{equation}
R_i = \frac{1}{{T}_b^i}, \quad 1 \leq i \leq N,
\label{eq29}
\end{equation}
where ${{T}_b^i}$ is the average back-off time of a vehicle on lane $i$, which is calculated as
\begin{equation}
{T}_{b}^{i} = \frac{W_0^i-1}{2} \times T_{slot}, \quad 1 \leq i \leq N,
\label{eq30}
\end{equation}
where $T_{slot}$ is the length of a time slot.

Collisions will hardly happen if the collision probability of a transmission is smaller than an upper bound $p_{UB}$. In this case, the transmission can be deemed to be collision-free \cite{AliMaatouk2020}. The collision probability for a vehicle on lane $i$, i.e., $p_i$, is the probability that at least one more vehicles is transmitting at the same time. Thus the relationship between the transmission probability $\tau_i$, the number of vehicles $n_i$ and the collision probability $p_i$ is calculated as
\begin{equation}
p_i = 1 - (1 - \tau_i)^{n_i - 1}\prod_{j=1,j \ne i}^N (1 - \tau_j)^{n_j}.
\label{eq31}
\end{equation}

Since $\tau_i$ is calculated according to Eq. \eqref{eq7}, a large $W_0^i$ can lead to a small $p_i$. According to Eq. \eqref{eq31}, the upper bound of the collision probability $p_{UB}$ corresponds to a lower bound of the minimum contention window $\overline{W}_0^{LB}$. Therefore, the network is collision-free if the minimum contention window of each vehicle is larger than $\overline{W}_0^{LB}$. Similar with \cite{AliMaatouk2020}, we consider the collision-free network in this paper, and thus the average service rate $H_i$ is calculated as
\begin{equation}
H_i = \frac{1}{T_s}, \forall i \in \{1, \ldots N \},
\label{eq32}
\end{equation}
where $T_s$ is the average successful transmission time of a packet \cite{Bianchi2000}.

\subsubsection{Average number of vehicles}
\
\newline
\indent
Next, the average number of vehicles in the network $N_v$ is derived. For simplicity, we first derive the average number of vehicles on lane $i$ in the network, as shown in Fig. \ref{fig3}. All vehicles on lane $i$ are moving with the same direction. It consists of $m_i$ complete platoons and an incomplete platoon. Note that the distance between the leader vehicles of two adjacent platoons denotes a complete platoon-interval. Then the coverage of the BS $R$ can be expressed as the summation of the average length of $m_i$ complete platoon-intervals and an incomplete platoon. In addition, the vehicles on lane $i$ are assumed to be with the same velocity, the same intra-platoon spacing and the same inter-platoon spacing.

The average number of vehicles on lane $i$ is expressed as

\begin{equation}
n_i = n_w^i + n_s^i,
\label{eq33}
\end{equation}
where $n_w^i$ and $n_s^i$ are the number of vehicles in complete platoons and incomplete platoons on lane $i$, respectively. Next, we will derive $n_w^i$ and $n_s^i$, respectively.

Given the number of complete platoons $m_i$ and the average number of vehicles in each platoon $n_p$, thus the number of vehicles in compete platoons on the lane is calculated as
\begin{equation}
n_w^i = m_i \times n_p.
\label{eq34}
\end{equation}

In order to facilitate the analysis, we define the platoon interval as the sum of the average length of a platoon and the following inter-platoon spacing. Thus there is $m_i$ platoon intervals on lane $i$, which is shown in Fig. \ref{fig3}. The length of a platoon interval is calculated as
\begin{equation}
D_w^i = d_w^i + r_s^i,
\label{eq35}
\end{equation}
where $d_w^i$ is the average length of a platoon and $r_s^i$ is the average distance between two consecutive platoons.

The average length of a platoon $d_w^i$ is composed of the length of $n_p$ vehicles and $n_p-1$ intra-platoon spacings, i.e.,
\begin{equation}
d_w^i = n_ps + (n_p-1)r_w^i = (n_p - 1)(s + r_w^i) + s,
\label{eq36}
\end{equation}
where $s$ is the length of a vehicle and $r_w^i$ is a intra-platoon spacing. Substituting Eq. \eqref{eq36} into Eq. \eqref{eq35}, we can observe that $D_w^i$ is a function of $n_p$, $r_w^i$ and $s$.

Since the number of complete platoons must be an integer, it can be obtained by rounding down the ratio of the coverage of the BS to the average length of a platoon, i.e.,
\begin{equation}
m_i = \left \lfloor \frac{R}{D_w^i} \right \rfloor.
\label{eq37}
\end{equation}

According to Eqs. \eqref{eq35}-\eqref{eq38}, $n_w^i$ is calculated as
\begin{equation}
n_w^i = \left \lfloor \frac{R}{(n_p - 1)(s + r_w^i) + s + r_s^i} \right \rfloor \times n_p.
\label{eq38}
\end{equation}

Next, the average number of vehicles in the incomplete platoon is derived. Since the transceiver is installed on the headstock of each vehicle, a vehicle enters the network as long as the headstock of the vehicle reaches the network. Therefore, the number of vehicles in an incomplete platoon is calculated by rounding up the ratio of the average length of the incomplete platoon $D_s^i$ to the sum of the vehicle length $s$ and an intra-platoon spacing $r_w^i$, i.e.,
\begin{equation}
n_s^i = \left \lceil \frac{D_s^i}{s + r_w^i} \right \rceil.
\label{eq39}
\end{equation}

The length of an incomplete platoon is calculated as the difference between the coverage area of the BS and the length of $m_i$ complete platoon intervals, i.e.,
\begin{equation}
D_s^i = R - m_i \times D_w^i = R - m_i \times (d_w^i + r_s^i).
\label{eq40}
\end{equation}

According to Eqs. \eqref{eq36}, \eqref{eq39} and \eqref{eq40}, $n_s^i$ is calculated as
\begin{equation}
n_s^i = \left \lceil \frac{R - m_i \times [(n_p - 1)(s + r_w^i) + s + r_s^i]}{s + r_w^i} \right \rceil.
\label{eq41}
\end{equation}

From Eqs. \eqref{eq38} and \eqref{eq41}, we can see that $n_w^i$ and $n_s^i$ are related with the average inter-platoon spacing $r_s^i$ and intra-platoon spacing $r_w^i$. Next, we further derive the relationship between $r_s^i$, $r_w^i$ and velocity $v_i$.

According to  \cite{JiaD2014}, the relationship between $r_w^i$ and $v_i$ is
\begin{equation}
{r_w^i} = \frac{{{r_0} + {v_i}{T_h}}}{{\sqrt {1 - {{\left( {\frac{{{v_i}}}{{{v_0}}}} \right)}^4}} }},
\label{eq42}
\end{equation}
where $r_0$ is the minimum intra-platoon spacing, $T_h$ is the time headway and $v_0$ is the maximum velocity of vehicles in the network.

The inter-platoon spacing should be greater than or equal to the intra-platoon spacing to prevent collisions between vehicles \cite{Naus2010}. Since the average distance between two consecutive platoons is $D_w^i$, the time difference between two consecutive platoons is calculated as
\begin{equation}
\Delta T_i = \frac{D_w^i}{v_i} = \frac{d_w^i + r_s^i}{v_i}, r_w^i \leq r_s^i.
\label{eq43}
\end{equation}

Thus the platoon arrival rate can be calculated as
\begin{equation}
\lambda_i = \frac{1}{\Delta T_i} = \frac{v_i}{d_w^i + r_s^i}, r_w^i \leq r_s^i.
\label{eq44}
\end{equation}

Note that according to Eqs. \eqref{eq36} and \eqref{eq44}, we can find that the platoon arrival rate $\lambda_i$ achieves the maximum value when the average inter-platoon spacing $r_s^i$ equals to the intra-platoon spacing $r_w^i$, i.e., $\lambda_i^{max} = \frac{v_i}{n_p(s + r_w^i)}$.

Substituting Eq. \eqref{eq36} into Eq. \eqref{eq44}, we can get the relationship between $r_s^i$, $r_w^i$ and $v_i$, i.e.,
\begin{equation}
r_s^i = \frac{v_i}{\lambda_i} - [(n_p - 1)(s + r_w^i) + s], \lambda_i \leq \frac{v_i}{n_p(s + r_w^i)}.
\label{eq45}
\end{equation}

Up to now, we have obtained the relationships between $r_w^i$, $r_s^i$ and $v_i$ according to Eqs. \eqref{eq42} and \eqref{eq45}, respectively. Substituting Eqs. \eqref{eq42} and \eqref{eq45} into Eqs. \eqref{eq38} and \eqref{eq41}, either $n_w^i$ and $n_s^i$ can be expressed as a function $v_i$, respectively. Thus we can further drive the average number of vehicles on lane $i$, i.e., $n_i$ as a function of $v_i$ according to Eq. \eqref{eq33}. As a result, we can obtain the relationship between the average number of vehicles in the network $N_v$ and $v_i$ by aggregating the average number of vehicles on each lane.

Conclusively, we have obtain the relationship between $N_v$ and $v_i$ and the relationship between $R_i$ and $W_0^i$. Moreover, $H_i$ can be determined by a given $T_s$. Thus according to Eq. \eqref{eq28}, given the velocities of vehicles on each lane, the average age of data in the network can be calculated according to the minimum contention window of vehicles on each lane.

\section{Optimization Problem and Solution}
\label{sec5}
In this section, we formulate a multi-objective optimization and adopt the multi-objective particle swarm optimization algorithm to solve it, thus the optimal minimum contention windows are derived.
\subsection{Optimization Objective}
In this paper, we adjust the minimum contention window of each vehicle to achieve two goals. First, it is preferable that the data transmitted by vehicles on different lanes have the similar size, i.e., the difference between fairness index $K_{index}^i$ for lane $i$ and $K_{index}$ should be as small as possible. Thus we have $N$ optimization objective functions as follows.

\textbf{Objective 1 to N:} Minimize the difference of the fairness index on different lanes \\
\begin{equation}
F_{K_i} = \min \left| K_{index}^i - K_{index} \right|,i \in [1, \ldots N].
\label{eq46}
\end{equation}

Another optimization goal is that the average age of data in the network is as small as possible, thus the $(N+1)$-th optimization objective function is

\textbf{Objective N+1:} Minimize the average age of data in the network \\
\begin{equation}
F_{age} = \min \overline{\Delta}.
\label{eq47}
\end{equation}

As mentioned above, the multi-objective optimization function is written as
\begin{equation}
\begin{aligned}
\boldsymbol{F} = & \begin{bmatrix}
      F_{K_{1}} \\
      F_{K_{2}} \\
      \vdots \\
      F_{K_{N}} \\
      F_{age}
\end{bmatrix}, \\
\qquad \qquad & S.t \\
& v_0' \leq v_i \leq v_0, \\
& W_0^{LB} \leq W_0^i \leq W_0^{UB}.
\label{eq48}
\end{aligned}
\end{equation}
where $v_0'$ and $v_0$ is the allowed minimum and maximum velocities of vehicles in the network, respectively. Thus the first constraint is to ensure that the velocity of vehicles in the network is within the allowable velocity range. Moreover, $W_0^{LB}$ is the lower bound of the minimum contention window to ensure the collision-free transmission, which is calculated according to Eqs. \eqref{eq7} and \eqref{eq31} under the given upper bound of the collision probability $p_{UP}$. Thus the transmissions in the network are collision-free condition if the minimum contention window of each vehicle is larger than $W_0^{LB}$ \cite{AliMaatouk2020}. $W_0^{UB}$ is the allowed upper bound of the minimum contention window defined in the 802.11 protocol. Thus the second constraint is used to ensure that the data transmissions in the network are collision-free and the minimum contention window is within the allowable range of 802.11 protocol.

According to the derivations in Section \ref{sec4}, given the velocities of vehicles on each lane, the fairness indexes on each lane and average age of data in the network in Eq. \eqref{eq48} can be designed through adaptively adjusting the minimum contention windows of each vehicle.

\subsection{Optimization Solution}
In this sub-section, we use the MoPSO algorithm to solve the multi-objective optimization function. A population includes $P$ particles. A particle denotes a set containing $N$ elements, i.e., the minimum contention windows of vehicles on $N$ lanes. The input of the algorithm is the set of the velocities of vehicles on each lane $\boldsymbol{v}$, and the output is the set of the optimal contention windows of vehicles on each lane $\boldsymbol{W}_0^{opt}$. The pseudocode of the algorithm is described as Algorithm \ref{al1}, where $\boldsymbol{p}_m^{old}$ and $\boldsymbol{p}_m^{new}$ are particle $m$ before and after update, respectively, $\boldsymbol{v}_m^{old}$ and $\boldsymbol{v}_m^{new}$ are the velocities of particle $m$ before and after update, respectively, $\boldsymbol{F}_m^{old}$ and $\boldsymbol{F}_m^{new}$ are the objective values of $\boldsymbol{p}_m^{old}$ and $\boldsymbol{p}_m^{new}$, respectively, $\boldsymbol{p}_m^{best}$ is the individual optimal solution of particle $m$, $\boldsymbol{g}_m^{best}$ is the global optimal solution of particle $m$, $t$ is the number of the iterations. The algorithm is composed of three procedures, i.e., initialization procedure, iteration procedure and optimization procedure. Next, we will introduce the procedures of the algorithm in detail.

\subsubsection{Initialization Procedure}
\
\newline
\indent
$P$ particles are initialized at first. Specifically, for each particle $m$, each element is initialized as an integer which is randomly selected within $[W_0^{LB},W_0^{UB}]$, thus $\boldsymbol{p}_m^{old}$ is generated. Then, the velocity of each element of particle $m$ is randomly selected within $[v_{min},v_{max}]$, thus $\boldsymbol{p}_m^{old}$ is generated. The objective value of each particle is calculated according to Eq. \eqref{eq48} under the input $\boldsymbol{v}$. The individual optimal solution of each particle is initialized as $\boldsymbol{p}_m^{best}=\boldsymbol{p}_m^{old}$ (lines 1-5).

Next, we determine the set of the non-inferior solutions, i.e., $pareto$ set, from all individual optimal solutions according to the $pareto$ dominance relationship \cite{Pareto_condition}. Specifically, a solution is a non-inferior solution if the solution dominates any other solution, i.e., the values of at least one objective value of the solution are better than the corresponding objective values of any other solution.

Then the grid method is used to divide the space of $pareto$ set into sub-regions. Thus we obtain $(mesh_{div})^N$ sub-regions with the same sizes by dividing $pareto$ set. Let $n_{pareto}$ be the number of individuals in the $pareto$ set, $\boldsymbol{p}_n^{pareto}$ be particle $n$ $(1\leq n \leq n_{pareto})$ in the $pareto$ set, $\boldsymbol{W}_0^{LB}$ be the lower bound of a particle where each element is ${W}_0^{LB}$, $\boldsymbol{id}_n$ be the index of the sub-region that particle $n$ belongs to, $\boldsymbol{id}_n$ can be calculated as
\begin{equation}
\begin{aligned}
\boldsymbol{id}_n = \left \lfloor (\boldsymbol{p}_n^{pareto} - \boldsymbol{W}_0^{LB}) \cdot \frac{n_{pareto}}{W_0^{UB} - W_0^{LB}} \right \rfloor.
\label{eq49}
\end{aligned}
\end{equation}
Hence we can map each particle of $pareto$ set to the index of the sub-region (lines 6-9).

Thus, the congestion degree of each sub-region, i.e., defined as the number of particles in each sub-region, is further obtained. The probability of particle $n$ which is used to select global optimal solution by the roulette gambling method is calculated according to congestion degree of the corresponding sub-region, i.e.,
\begin{equation}
p_n = \frac{1}{(n_g)^\alpha},
\label{eq50}
\end{equation}
where $n_g$ is the number of particles in the sub-region, $\alpha$ is a given constant. Then the probability of each particle in the $pareto$ set is normalized in Line 10.

Next, the roulette gambling method is adopted to choose the global optimal solution of each particle of the population. Specifically, each particle of the population generates a random probability from $[0,1]$. Denote $p_m^r$ as the generated random probability of particle $m$, and $p_n^{sum}$ as the summation of the normalized probabilities from particle $1$ to $n$ ${(1 \leq n \leq n_{pareto})}$ in the $pareto$ set. For particle $m$, $\boldsymbol{g}_m^{best}$ is selected through comparison as follows. At the beginning $n=1$, if $p_m^r < p_n^c$, particle $n$ in the $pareto$ set is selected as the global optimal solution of particle $m$, i.e., $\boldsymbol{g}_m^{best}$; otherwise, we will go to the next particle, i.e., $n=n+1$, and repeat the above comparison process until the global optimal solution of particle $m$, i.e., $\boldsymbol{g}_m^{best}$, is selected (line 11).

\begin{algorithm}
  \caption{Multi-objective Particle Swarm Optimization Algorithm}
  \KwIn{$\boldsymbol{v}$}
  \KwOut{$\boldsymbol{W}_0^{opt}$}
  \label{al1}
  \For{each $m \in P$}
  {
    randomly generate $\boldsymbol{p}_m^{old}$ from $[W_0^{LB},W_0^{UB}]$\;
    randomly generate $\boldsymbol{v}_m^{old}$ from $[v_{min},v_{max}]$\;
    calculate the value of objective function $\boldsymbol{F}_m^{old}$\;
    initialize individual optimal solution $\boldsymbol{p}_m^{best}$\;
  }
  initialize the $pareto$ set\;
  divide the space of $pareto$ set by grid method\;
  \For{$(n \in pareto)$}
  {
    calculate the sub-region index $\boldsymbol{id}_n$\;
  }
  calculate normalized probability of each particle in $pareto$ set\;
  roulette gambling method to choose $\boldsymbol{g}_m^{best}$ for each $m \in P$\;
  initialize $t = 0$\;
  \While{$t \leq c$}
  {
    \For{$(m \in P)$}
    {
      update the velocity of particle {$\boldsymbol{v}_m^{new} = \omega \times \boldsymbol{v}_m^{old} + c_1(\boldsymbol{p}_m^{best} - \boldsymbol{p}_m^{old}) + c_2(\boldsymbol{g}_m^{best} - \boldsymbol{p}_m^{old})$}\;
      update particle {$\boldsymbol{p}_m^{new} = \boldsymbol{p}_m^{old} + \boldsymbol{v}_m^{new}$}\;
     update the value of objective function $\boldsymbol{F}_m^{new}$\;
      \If{$\boldsymbol{p}_m^{new}$ dominates $\boldsymbol{p}_m^{old}$}
      {
        $\boldsymbol{p}_m^{old}$ = $\boldsymbol{p}_m^{new}$\;
      }
      ${\boldsymbol{v}_m^{old} = \boldsymbol{v}_m^{new}}$\;
      \If{$\boldsymbol{p}_m^{old}$ dominates $\boldsymbol{p}_m^{best}$}
      {
        $\boldsymbol{p}_m^{best}$ = $\boldsymbol{p}_m^{old}$\;
      }
    }
    update $pareto$ solution set\;
    \For{$(n \in pareto)$}
    {
        calculate the sub-region index $\boldsymbol{id}_n$\;
    }
    calculate the normalized probability\;
    roulette gambling update $\boldsymbol{g}_m^{best}$ for each $m \in P$\;
    $t = t + 1$\;
  }
  initialize $\overline{\Delta} = +\infty$\;
  \For{$(n \in pareto)$}
  {
   \If{for each $i \in [1,...,N]$, $F_{K_i} \leq K_{bound}$}
   {
     \If{$F_n^{age} \leq \overline{\Delta}$}
      {
        $\overline{\Delta} = F_n^{age}$, $\boldsymbol{W}_0^{opt} = \boldsymbol{W}_0^n$\;
      }
    }
  }
  return $\boldsymbol{W}_0^{opt}$\;
\end{algorithm}

\subsubsection{Iteration Procedure}
\
\newline
\indent
We will run the iteration procedure repeatedly to update the $pareto$ set.

In each iteration, the velocity of each particle of the population is updated as
\begin{equation}
\begin{aligned}
\boldsymbol{v}_m^{new} = \omega \times \boldsymbol{v}_m^{old} + c_1(\boldsymbol{p}_m^{best} - \boldsymbol{p}_m^{old})+c_2(\boldsymbol{g}_m^{best} - \boldsymbol{p}_m^{old}).
\label{eq51}
\end{aligned}
\end{equation}
In Eq. \eqref{eq51}, if the updated velocity of an element in a particle is larger than $v_{max}$ or smaller than $v_{min}$, it would be set to $v_{max}$ or $v_{min}$, respectively (lines 12-15).

Then each particle is updated as
\begin{equation}
\begin{aligned}
\boldsymbol{p}_m^{new} = \boldsymbol{p}_m^{old} + \boldsymbol{v}_m^{new}.
\label{eq52}
\end{aligned}
\end{equation}
In Eq. \eqref{eq52}, if an element of a updated particle is larger than $W_0^{UB}$ or smaller than $W_0^{LB}$, it would be set to $W_0^{UB}$ or $W_0^{LB}$, respectively. Then the objective value of each updated particle is calculated according to Eq. \eqref{eq48} under the input $\boldsymbol{v}$ (lines 16-17).

Then each particle is updated according to the $pareto$ dominance relationship \cite{Pareto_condition}. Specifically, for individual $m$, if $\boldsymbol{p}_m^{new}$ dominates $\boldsymbol{p}_m^{old}$, $\boldsymbol{p}_m^{old}$ is updated as $\boldsymbol{p}_m^{new}$; otherwise, $\boldsymbol{p}_m^{old}$ is not updated. Then $\boldsymbol{v}_m^{old}$ is updated as $\boldsymbol{v}_m^{new}$ (lines 18-20).

Then each individual optimal solution is updated according to the $pareto$ dominance relationship. Specifically, if $\boldsymbol{p}_m^{old}$ deminates $\boldsymbol{p}_m^{best}$, $\boldsymbol{p}_m^{best}$ is updated as $\boldsymbol{p}_m^{old}$. Otherwise, $\boldsymbol{p}_m^{best}$ is not updated (lines 21-22).


Next, the $pareto$ set is updated through comparing the dominance relationship of solutions in the new set which is the combination of the updated individual optimal solutions and the current $pareto$ set. Afterwards, the grid method is adopted to divide the space of the updated $pareto$ set into sub-regions. Then the index of the sub-region for each solution is calculated according to Eq. \eqref{eq49}. Next, the congestion degree of each sub-region is calculated. And thus we can derive the probability of each particle in the $pareto$ set according to Eq \eqref{eq50} and further obtain the normalized probability of each particle in the $pareto$ set. Then, the global optimal solution of each particle of the population is updated according to the roulette gambling method. Afterwards, we will go to the next iteration (lines 23-28).
\begin{table}\footnotesize
\caption{Related parameters values}
\label{tab3}
\centering
\begin{tabular}{|c|c|c|c|}
\hline
\textbf{Parameter} &\textbf{Value} &\textbf{Parameter} &\textbf{Value}\\
\hline
$R$ & $200m$ & $N$ & $2$ \\
\hline
$r_{0}$ & $2m$ & $n_p$ & 2 \\
\hline
$v_0'$ & $20m/s$ & $v_0$ & $30m/s$ \\
\hline
$s$ & $5m$ & $T_h$ & $1.6s$ \\
\hline
$\overline{v}$ & $25m/s$ & $\overline{W}$ & 128 \\
\hline
$T_{slot}$ & $50\mu s$ & $T_s$ & $8972 \mu s$ \\
\hline
$P$ & $200$ & $c$ & $100$ \\
\hline
$\omega$ & $0.8$ & $c_1$ & $0.9$ \\
\hline
$c_2$ & $1.8$ & $mesh_{div}$ & $10$ \\
\hline
$W_0^{LB}$ & $64$ & $W_0^{UB}$ & $256$ \\
\hline
$v_{min}$ & $-1.5$ & $v_{max}$ & $1.5$ \\
\hline
$\alpha$ & $3$ & $K_{bound}$ & $0.005$ \\
\hline
$p_{UP}$ & $0.24$ & $ $ & $ $ \\
\hline
\end{tabular}
\vspace{-0.3cm}
\end{table}

\subsubsection{Optimization Procedure}
\
\newline
\indent
In the next process, the optimal solution is selected from the $pareto$ set to jointly optimize
the access fairness and freshness of data, i.e., satisfy the following two conditions,
\begin{itemize}
\item[(1)] The differences between the fairness index of each lane and the average fairness index of the network $K_{index}$ are all small enough, i.e., less than a small threshold $K_{bound}$.
\item[(2)] Average age of data in the network is as small as possible when condition (1) is satisfied.
\end{itemize}

The average age of data $\overline{\Delta}$ is first set as a very large value which is denoted as infinity. Let $F_n^{age}$ be the average age of data calculated according to the particle $n$ in the $pareto$ set, i.e., $\boldsymbol{p}_n^{pareto}$. The optimal minimum contention windows are obtain through the following comparison for each particle in the $pareto$ set. Specifically, for particle $n$ in the $pareto$ set, if the fairness differences of each lane are all smaller than $K_{bound}$ and the age $F_n^{age}$ is smaller than $\overline{\Delta}$, $\overline{\Delta}$ is updated as $F_n^{age}$ and $\boldsymbol{W}_0^{opt}$ is set as $\boldsymbol{p}_n^{pareto}$. Finally, the optimal solution $\boldsymbol{W}_0^{opt}$ is achieved when the iteration is finished.

\subsection{Algorithm Complexity}
The complexity of the algorithm is mainly determined by the iterative number of algorithm $c$, the population size $P$ and the number of multi-objective equations simultaneously. According to the multi-objective optimization equations at the beginning of this section, it is composed of $N$ fairness equations and one network average age equation. Therefore, the algorithm complexity is about ${\mathcal{O}(cP(N+1))}$

\begin{figure}[htbp]
\centering
\includegraphics[width=\linewidth, scale=1.00]{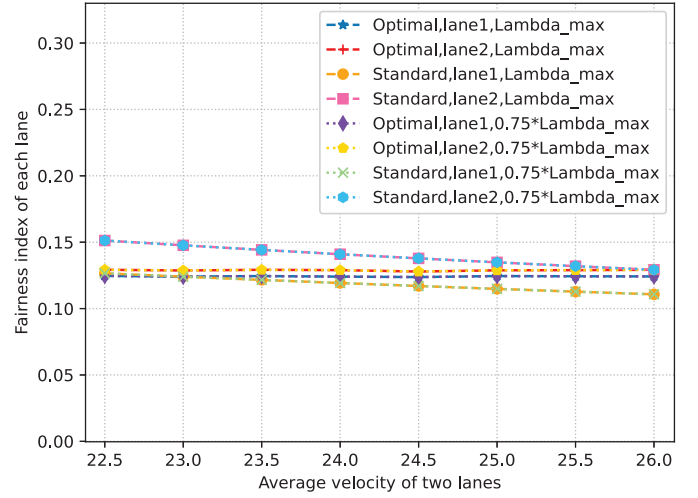}
\caption{Fairness index of each lane vs average velocity}
\label{fig4}
\vspace{-0.3cm}
\end{figure}
\begin{figure}[htbp]
\centering
\includegraphics[width=\linewidth, scale=1.00]{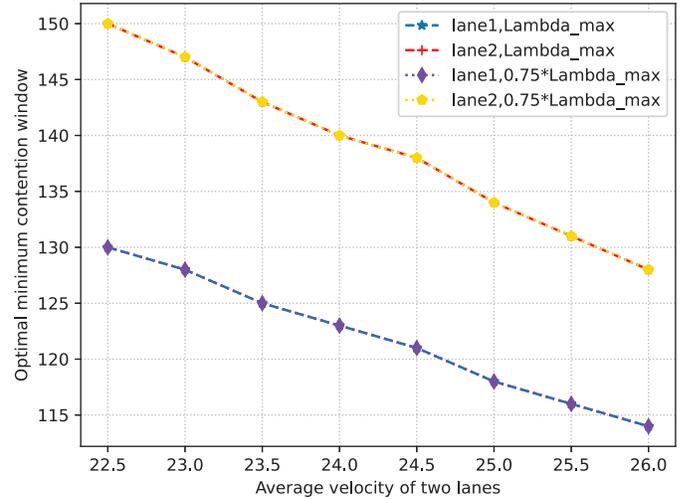}
\caption{Optimal minimum contention window vs average velocity}
\label{fig5}
\vspace{-0.4cm}
\end{figure}
\section{Numerical Simulation And Analysis}
\label{sec6}
In this section, we evaluate performance of the scheme by extensive simulation experiments\footnote{Code to replicate the numerical experiments here presented can be found at https://github.com/qiongwu86/Fairness-Data-Freshness/branches} and discuss the results in detail. The simulation is based on python $3.6$. A scenario in a two-lane one-way highway is considered. The vehicle velocity ranges from $20 m/s$ to $30 m/s$ according to the limit of American highways \cite{Vehicles2015}. The velocity difference between the two lanes keeps $4 m/s$. The average minimum contention window $\overline{W}$, average successful transmission time $T_s$ and the length of a time slot $T_{slot}$ are set according to the 802.11 protocol \cite{802.11}. The up bound of the collision probability $p_{UP}$ is set according to \cite{AliMaatouk2020}. Given $p_{UP}$, $W_0^{LB}$ is calculated according to Eqs. \eqref{eq7} and \eqref{eq31}. Table \ref{tab3} lists the parameters in the simulation.

Fig. \ref{fig4} shows the relationships between the fairness index of each lane and the average velocity of two lanes under platoon arrival rates $\lambda_{max}$ and $0.75\lambda_{max}$. Note that the minimum contention windows are the optimal value and the standard value defined in the 802.11 protocol, i.e., 128, respectively. We can see that when the optimal minimum contention windows are adopted, the fairness indexes of lane 1 and lane 2 almost keep stable under different average velocities and platoon arrival rates. It means that vehicles with different velocity can fairly access the channel with the optimal minimum contention windows. Moreover, we can see that when the standard contention window is adopted, the fairness index gradually decreases with the increase of average velocity. This is because the amount of data transmitted by vehicles with high velocities is smaller than that of vehicles with low velocity. Moreover, it can be seen that the fairness index is almost the same under different platoon arrival rates since the fairness index is only related to the vehicle velocity and the minimum contention window according to Eq. \eqref{eq8}.

\begin{figure*}
  \centering
  \subfigure[]{
    \begin{minipage}{5.5cm}
    \includegraphics[scale=0.4]{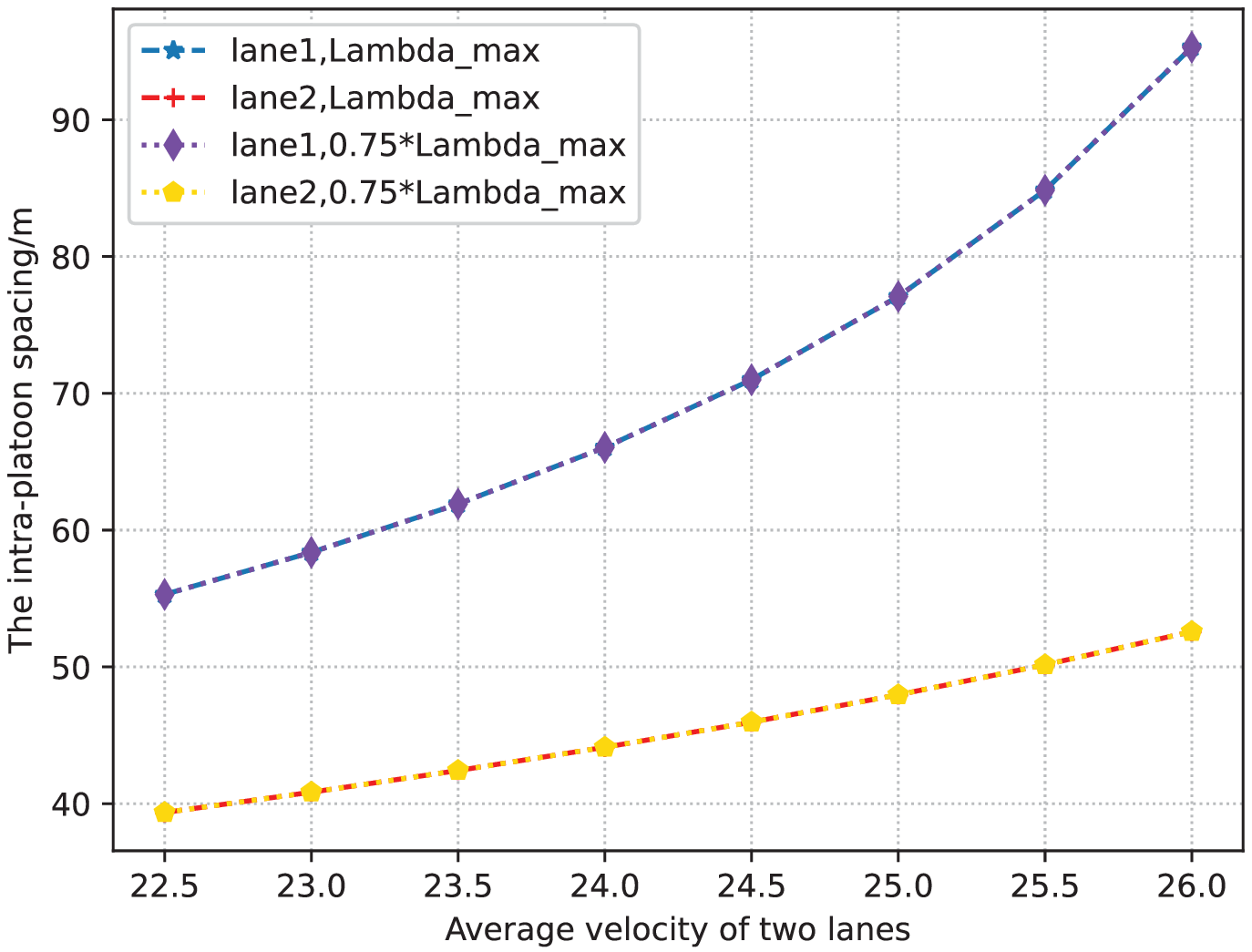}
    \end{minipage}}
  \subfigure[]{
    \begin{minipage}{5.5cm}
    \includegraphics[scale=0.4]{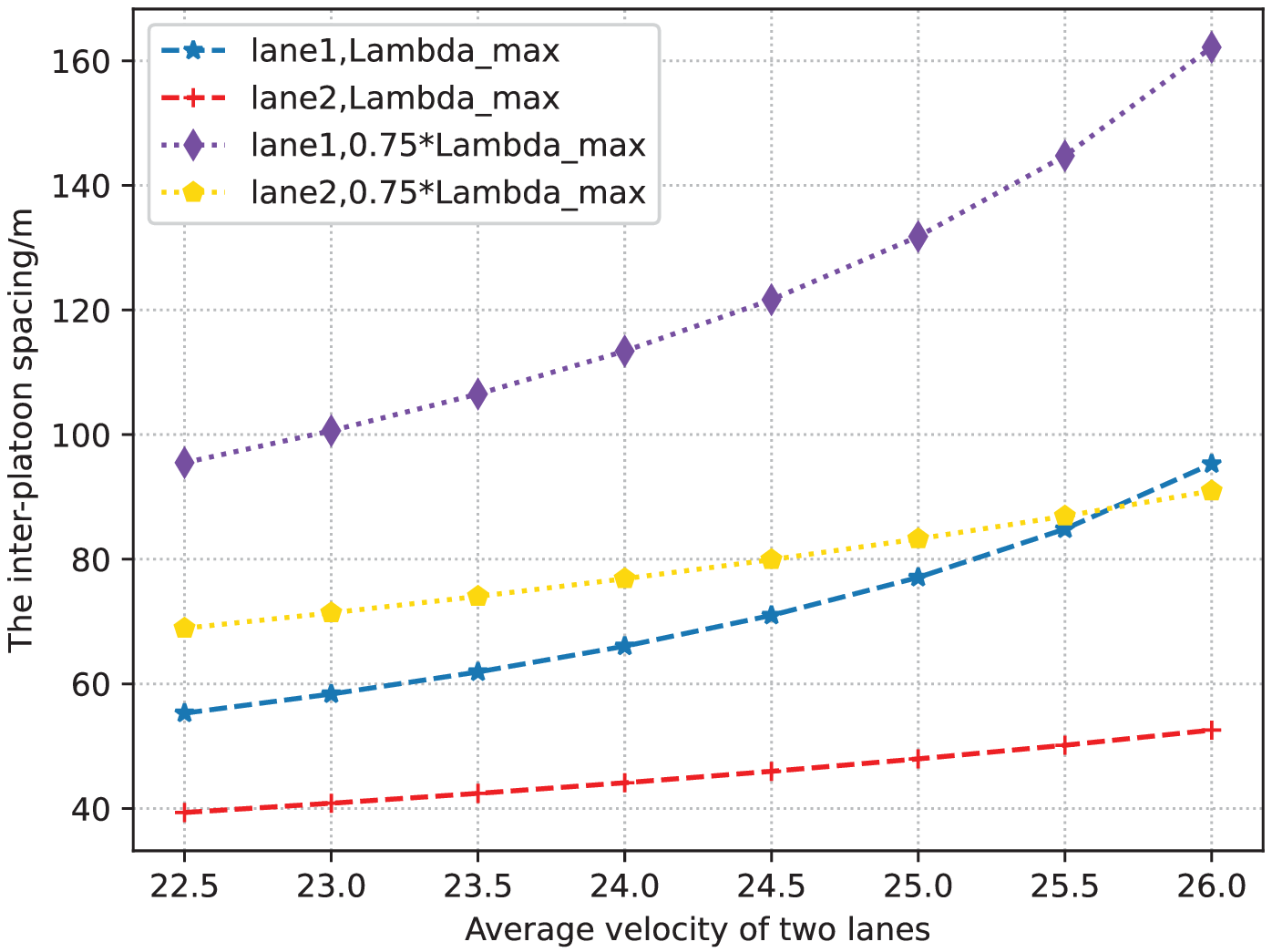}
    \end{minipage}}
  \subfigure[]{
    \begin{minipage}{5.5cm}
    \includegraphics[scale=0.4]{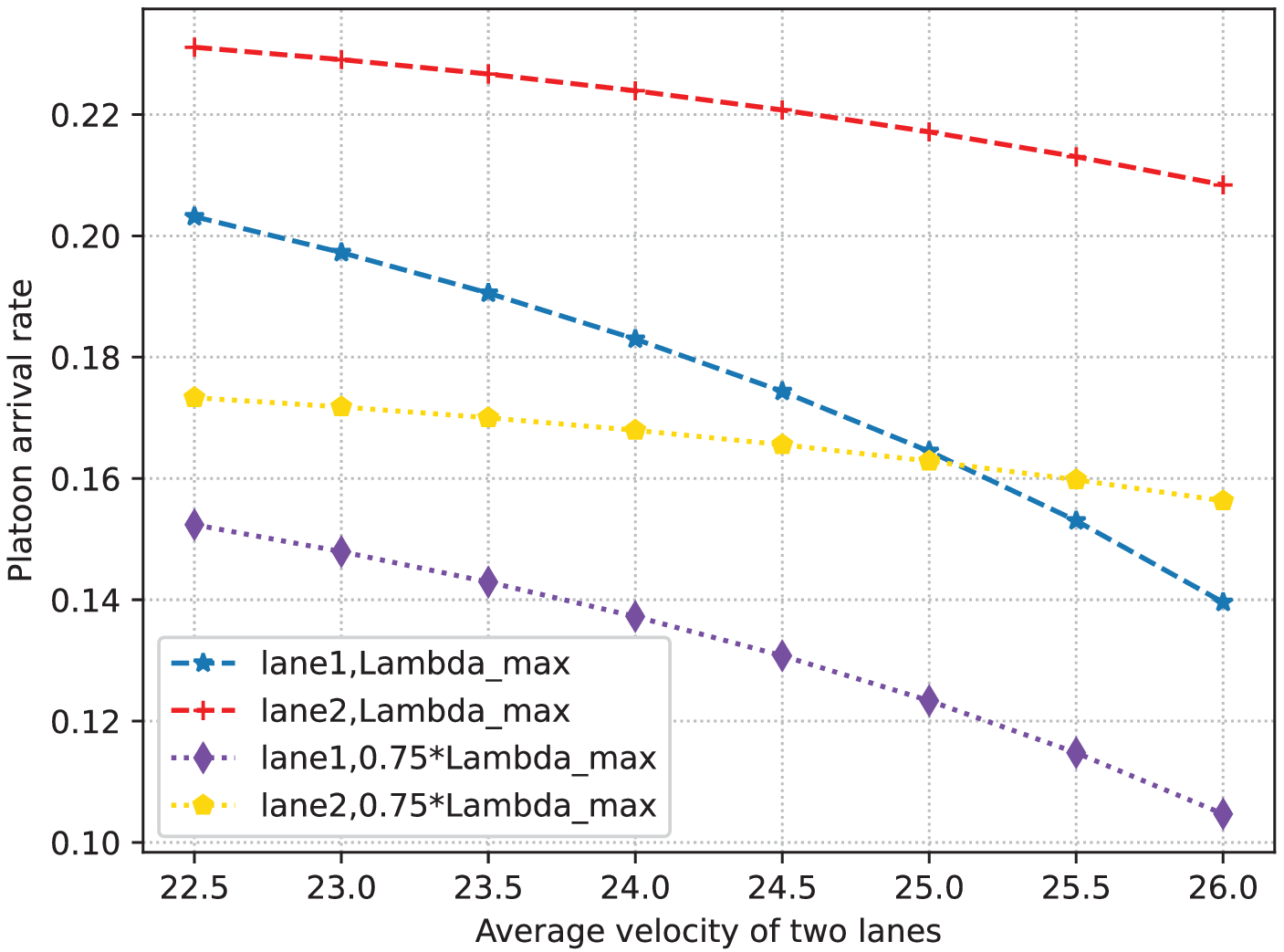}
    \end{minipage}}
    \caption{Platoon characteristics. (a) Intra-platoon spacing; (b) Inter-platoon spacing; (c) Platoon arrival rate.}
    \label{fig6}
    \vspace{-0.4cm}
\end{figure*}
\begin{figure*}
  \centering
  \subfigure[]{
    \begin{minipage}{5.5cm}
    \includegraphics[scale=0.4]{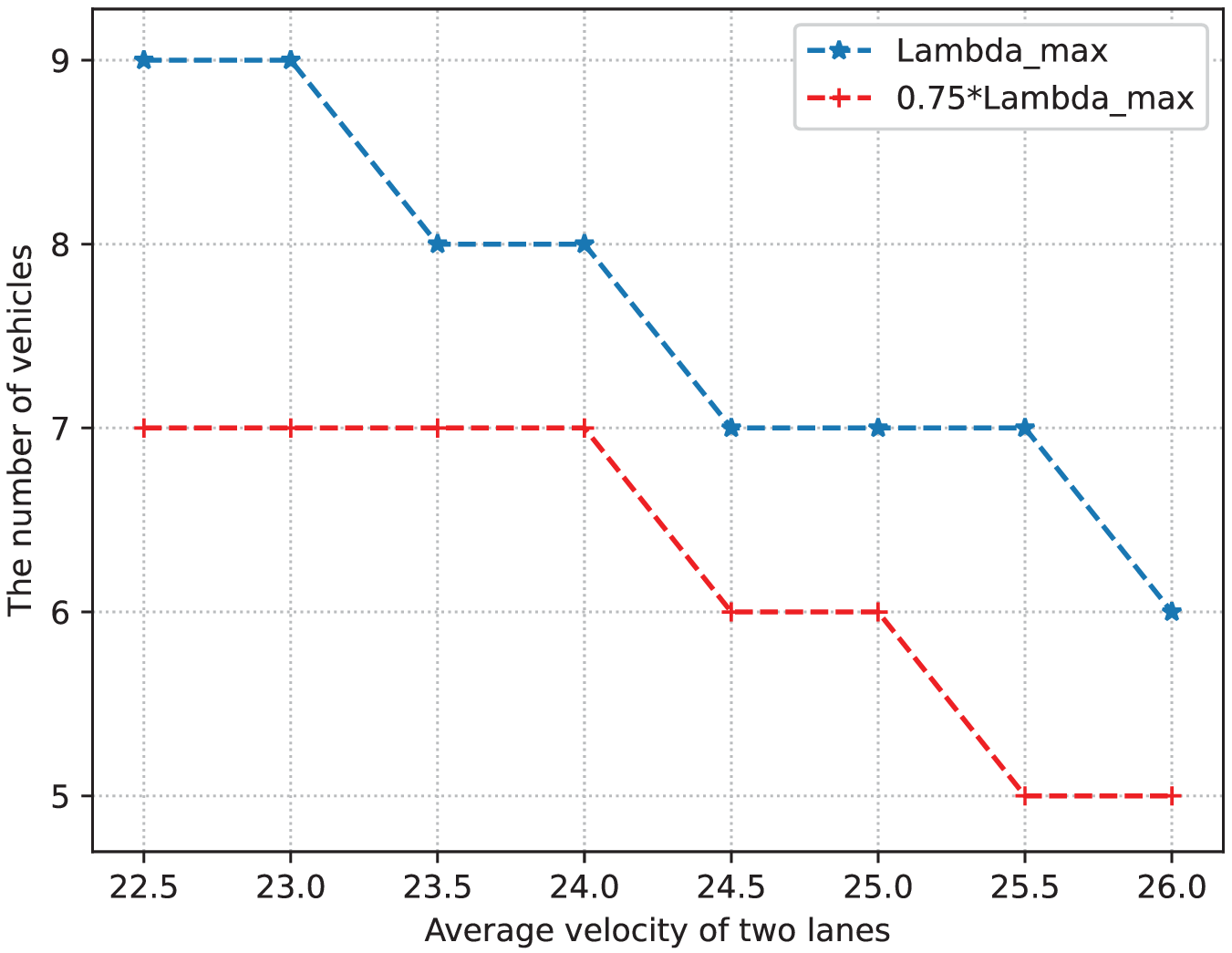}
    \end{minipage}}
  \subfigure[]{
    \begin{minipage}{5.5cm}
    \includegraphics[scale=0.4]{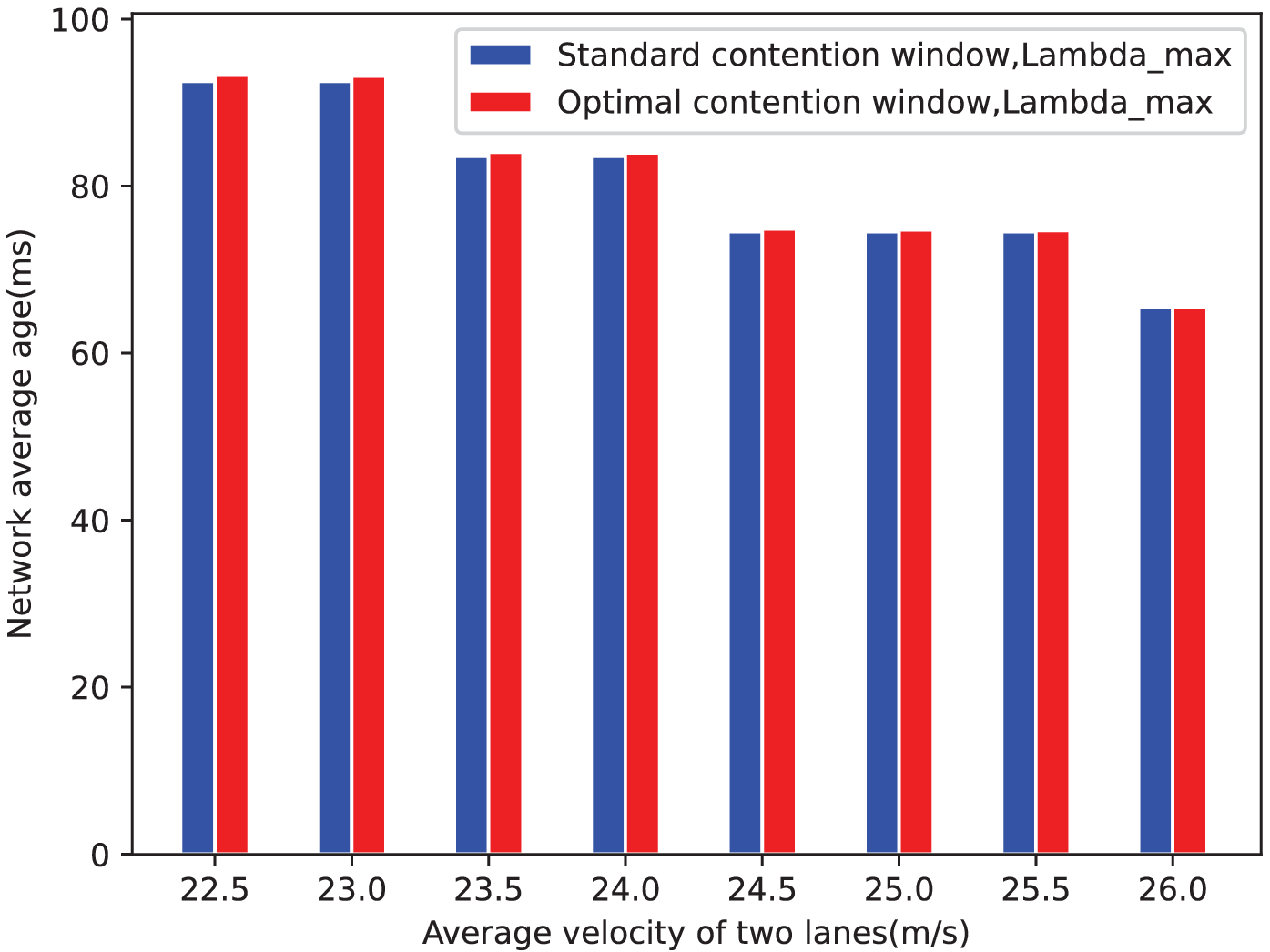}
    \end{minipage}}
  \subfigure[]{
    \begin{minipage}{5.5cm}
    \includegraphics[scale=0.4]{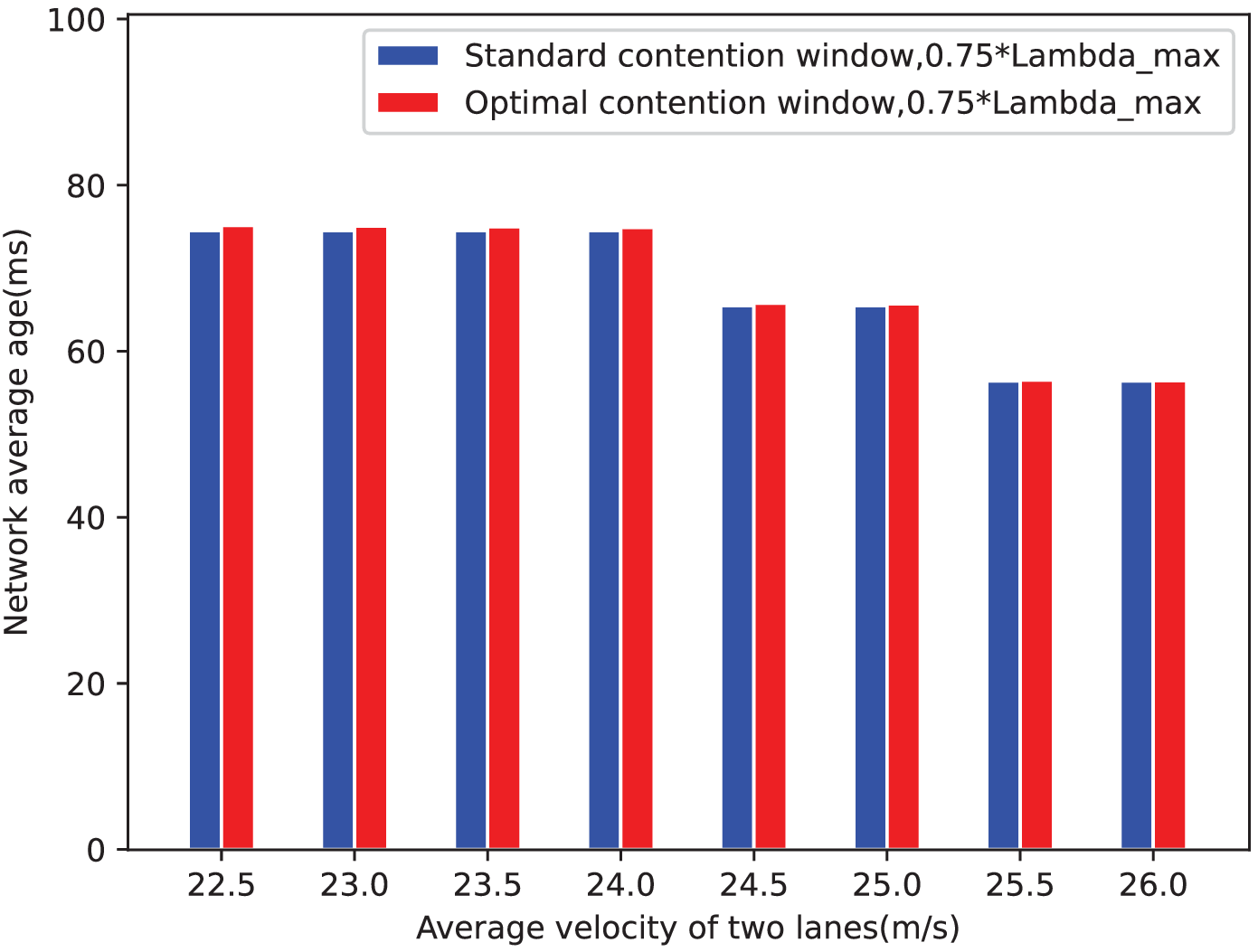}
    \end{minipage}}
    \caption{Total vehicles number and average age. (a) Number of vehicles; (b) Average age of data in network under $\lambda_{max}$; (c) Average age of data in network under $0.75\lambda_{max}$.}
    \label{fig7}
    \vspace{-0.4cm}
\end{figure*}

Fig. \ref{fig5} shows the relationship between the optimal minimum contention window and average velocity under different platoon arrival rates. We can see that the optimal minimum contention window decreases as the average velocity increases. This is because that vehicles with different velocities adjust the minimum contention window adaptively to keep the fairness index constant, which is consistent with Eq. \eqref{eq8}.

Figs \ref{fig6}-(a), (b) and (c) show the impact of the average velocity of the two lanes on the network performance in terms of intra-platoon spacing, inter-platoon spacing and platoon arrival rate. In Fig. \ref{fig6}-(a), the intra-platoon spacing is the same under different arrival rates on the same lane since the intra-platoon spacing is only related to the velocity of the vehicle, which is consist with Eq. \eqref{eq42}. Moreover, the intra-platoon spacing gradually increases as the average velocity increases. In Fig. \ref{fig6}-(b), the inter-platoon spacing achieves the minimum value when the platoon arrival rate is set as $\lambda=\lambda_{max}$, which is consistent with Eq. \eqref{eq45}. Moreover, the inter-platoon spacing is smaller with a larger platoon arrival rate. This is because that given the vehicle velocity and intra-platoon spacing, the inter-platoon spacing in Eq. \eqref{eq45} increases with the decrease of platoon arrival rate. In addition, the inter-platoon spacing gradually increases as the average velocity increases, which keeps consistent with Eq. \eqref{eq45}. In Fig. \ref{fig6}-(c), the platoon arrival rate decreases as the average velocity increases. This is because that the inter-and intra-platoon spacing increase with the increase of the vehicle velocity, thus leading to the decrease of the platoon arrival rate, which is consistent with Eq. \eqref{eq44}.

Fig. \ref{fig7}-(a) shows the impact of the average velocity on the number of vehicles in the network under different platoon arrival rates. Figs (b) and (c) show the impact of the average velocity on the average age of data in the network under different platoon arrival rates. In Fig. \ref{fig7}-(a), the number of vehicles in the network gradually decreases as the average velocity increases. This is because the inter-and intra-platoon spacing are gradually increasing as the average velocity increases, which results in a gradual decrease of the number of vehicles on each lane. In addition, for the same lane, the number of vehicles under platoon arrival rate $0.75\lambda_{max}$ is significantly less than that under $\lambda_{max}$. In Figs (b) and (c), the average age of data in the network under the optimal minimum contention windows is almost the same as that under the standard minimum contention window. However, the access of vehicles with different velocities under the optimal minimum contention windows is fair but the standard minimum contention window cannot achieve the fair access. In addition, the age decreases obviously with the increase of average velocity. This is because the number of vehicles in the network gradually decreases as the average velocity increases, which is consistent with Fig. \ref{fig7}-(a). In other words, the number of vehicles will have a significant impact on the data freshness in the network.
\begin{figure}[htbp]
\centering
\includegraphics[width=\linewidth, scale=1.00]{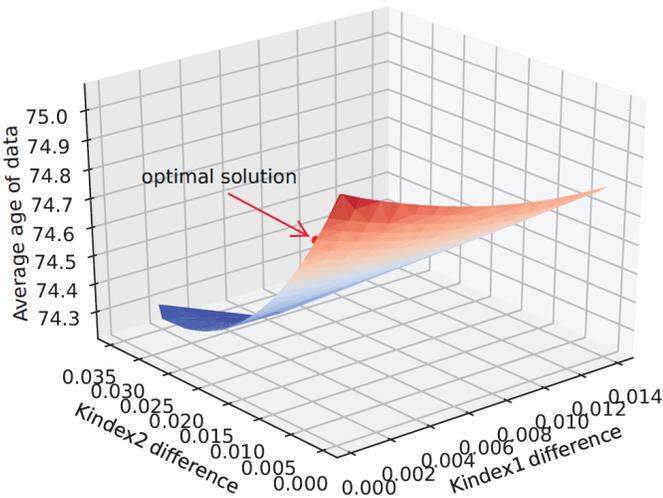}
\caption{Pareto surface}
\label{fig8}
\vspace{-0.4cm}
\end{figure}

Fig. \ref{fig8} shows the $pareto$ front surface through fitting the solutions in the $pareto$ set  under platoon arrival rate $\lambda_{max}$. The vehicle velocities on lane $1$ and lane $2$ are $26.5m/s$ and $22.5m/s$, respectively, thus the average velocity of two lanes is $24.5m/s$. It can be seen that fairness difference for each lane gradually increases when the average age of data in the network decreases. On the contrary, the average age of data in the network gradually increases when the fairness difference for each lane gradually decreases. Therefore, it is difficult to achieve the fairness access of the network and minimize the average age of data in the network at the same time. The optimal solution is shown by the red dot in Fig. \ref{fig8}. It is seen that the fairness difference is relatively small and the age is not very large for the optimal solution. The similar $pareto$ front surface can be obtained under different average velocities.

It is noted that when the inertia factor is 0.8 and the speed factors are 0.9 and 1.8, the global optimal solution can be obtained for most of the velocities, while the solution may fall in the local optima for only few velocities, i.e., 22.5m/s and 25.5m/s under the maximum arrival rate of the platoon, and 22.5m/s under the 0.75 maximum arrival rate. To solve the local optimal problem, we increase the inertia factor exponentially with increasing the number of iterations for the velocities. Specifically, the inertia factor is small in the early stage of the algorithm, thus the individual of the population can explore the solution efficiently. In this way, a relatively better solution can be obtained for the velocities.

As mentioned above, we can obtain the following conclusion from Fig. \ref{fig8}.
\begin{itemize}
\item
It is difficult to achieve the fairness access and minimize the average age of data in the network simultaneously given the velocity of vehicles.
\end{itemize}
This can be explained according to Lemma \ref{lemma1} and Lemma \ref{lemma2}. Based on Lemma \ref{lemma1}, the fairness access is achieved when ${v_1W_0^1 \approx v_2W_0^2 \approx \dots \approx v_NW_0^N \approx \overline{v}\overline{W}_0}$, while based on Lemma \ref{lemma2} the minimize the average age of data in the network is achieved when ${W_0^1 \approx W_0^2 \approx \dots \approx W_0^N \approx \overline{W}_0}$. However, the velocities of vehicles on different lanes are different, thus ${v_1 \neq v_2 \neq \dots v_N}$. Therefore, when the fairness access is achieved, ${W_0^1 \neq W_0^2 \neq \dots W_0^N}$. In this case, the minimum the average age of data in the network is difficult to be achieved. Therefore, it is difficult to achieve the fairness access and minimize the average age of data in the network simultaneously given the velocity of vehicles.

Moreover, we can obtain the following conclusion from Figs. \ref{fig4}, \ref{fig7}-(b) and \ref{fig7}-(c).
\begin{itemize}
\item The optimal minimum contention windows can achieve relative fair access while having the similar data freshness in the network as compared to standard minimum contention window.
\end{itemize}
This is because that the fairness index of a vehicle is determined by the minimum contention window of the vehicle according to Eq. \eqref{eq8} when the velocity of the vehicle is given. Thus, the fairness access can be achieved when the windows of vehicles are adjusted to be optimal contention windows. However, the standard minimum contention window cannot achieve the fairness access. Therefore, we can get part of the conclusion, the optimal contention windows can achieve the fairness access as compared to the standard minimum contention window. In addition, the data freshness in the network is related with $R_k/R_j$ according to Eq. \eqref{eq56} when average service rate of vehicles in the network $H_s$ is given, while $R_k$ and $R_j$ are inversely proportional to the minimum contention window $W_0^k-1$ and $W_0^j-1$ according to Eq. \eqref{eq59}, respectively, which leads to the data freshness in the network is related with $(W_0^k-1)/(W_0^j-1)$.Thus, the change of minimum contention windows has small impact on the data freshness in the network. In this case, the data freshness in the network under the optimal contention windows is similar with the data freshness under the standard minimum contention window. Therefore, we can get another part of the conclusion, i.e., the optimal contention windows can achieve similar data freshness in the network as compared to standard minimum contention window.

\section{Conclusion}
\label{sec7}
In this paper, we considered both fairness and freshness of data in the MEC-assisted platooning network, and designed a multi-objective optimal scheme to jointly optimize the access fairness and data freshness through adjusting the minimum contention windows of vehicles with different velocities. We formulated a multi-objective optimization problem and adopted the MoPSO algorithm to solve it and also got the optimal minimum contention windows of vehicles. According to theoretical analysis and the simulation results, we can get the conclusions as follows:
\begin{itemize}
\item Due to the minimum contention windows of vehicles on different lanes are difficult to be approximated under various velocities, it is difficult to achieve the fairness access and minimize the average age of data in the network simultaneously given the velocities of vehicles.
\item Since the fairness index and data freshness are both functions of the minimum contention window, and the change of the minimum contention window has a much greater impact on the fairness index than the data freshness, the optimal minimum contention windows can achieve fair access while having the similar data freshness in the network as compared to standard minimum contention window.
\end{itemize}

In the future, optimizing both the minimum contention window and average service rate of each vehicle may be an important direction so as to achieve the fairness access if the lane changing among multiple lanes with the same moving direction is allowed.

\appendix
\label{appxA}
In the appendix, we prove that the minimum average age of information of the network is achieved when all vehicles in the network adopt the almost same minimum contention window under the condition of the average minimum contention window of network is given. Since the summation of number of vehicles on each lane is equal to the number of vehicles in the network, Eq. \eqref{eq28} can be written as
\begin{equation}
\overline{\Delta} = \frac{1}{N_v} (\sum_{j=1}^{N_v} \frac{C(\boldsymbol{R}) }{R_j} + \frac{N_v }{C(\boldsymbol{R})} \sum_{j=1}^{N_v} \frac{R_j}{{H_j}^2}).
\label{eq53}
\end{equation}

Since $T_{slot}$ is much smaller than $T_s$, $R_k$ is much larger than $H_k$ for vehicle $k$ according to Eqs. \eqref{eq29}, \eqref{eq30} and \eqref{eq32}, thus $\frac{R_k}{H_k}$ is much larger than 1. Therefore, Eq. \eqref{eq25} can be approximated as
\begin{equation}
C(\boldsymbol{R})  \approx \sum_{k = 1}^{N_v} \frac{R_k}{H_k},
\label{eq54}
\end{equation}
Substituting Eq. \eqref{eq54} into \eqref{eq53}, Eq. \eqref{eq53} can be rewritten as
\begin{equation}
\overline{\Delta} \approx \frac{1}{N_v} (\sum_{j=1}^{N_v} \frac{\sum_{k = 1}^{N_v} \frac{R_k}{H_k} }{R_j} + \frac{N_v }{\sum_{k = 1}^{N_v} \frac{R_k}{H_k}} \sum_{j=1}^{N_v} \frac{R_j}{{H_j}^2}).
\label{eq55}
\end{equation}

Assuming all vehicles in the network are given the same average service rate, i.e., $H_i = H_s=\frac{1}{T_s}, {\forall i \in [1,2, \dots N_v]}$, Eq. \eqref{eq55} can be further simplified as
\begin{equation}
\overline{\Delta} \approx \frac{1}{H_s} (\sum_{j=1}^{N_v} \frac{\sum_{k = 1}^{N_v} R_k}{R_j} + 1).
\label{eq56}
\end{equation}

In Eq. \eqref{eq56}, ${\frac{\sum_{k = 1}^{N_v} R_k}{R_j}}$ can be written as ${(1+\sum_{k = 1,k \neq j}^{N_v}\frac{R_k}{R_j})}$, thus ${\sum_{j=1}^{N_v}\sum_{k = 1}^{N_v}\frac{R_k}{R_j}}$ can be rewritten as ${(N_v+\sum_{j=1}^{N_v}\sum_{k=1,k \neq j}^{Nv}\frac{R_k}{R_j})}$, here
\begin{equation}
\begin{aligned}
\sum_{j=1}^{N_v}\sum_{k=1,k \neq j}^{Nv}\frac{R_k}{R_j}&=\frac{R_2}{R_1}+\frac{R_3}{R_1}+\frac{R_4}{R_1}+ \dots +\frac{R_{N_v}}{R_1} \\
&+ \frac{R_1}{R_2}+\frac{R_3}{R_2}+\frac{R_4}{R_2}+ \dots +\frac{R_{N_v}}{R_2} \\
&\dots \\
&+ \frac{R_1}{R_{N_v}}+\frac{R_2}{R_{N_v}}+\frac{R_3}{R_{N_v}}+ \dots +\frac{R_{N_v-1}}{R_{N_v}},
\label{eq57}
\end{aligned}
\end{equation}
after reorganizing, we can get
\begin{equation}
\begin{aligned}
&\sum_{j=1}^{N_v}\sum_{k=1,k \neq j}^{Nv}\frac{R_k}{R_j}\\
&=(\frac{R_2}{R_1}+\frac{R_1}{R_2})+(\frac{R_3}{R_1}+\frac{R_1}{R_3})+\dots+(\frac{R_{N_v}}{R_1}+\frac{R_1}{R_{N_v}}) \\
&+(\frac{R_3}{R_2}+\frac{R_2}{R_3})+\dots+(\frac{R_{N_v}}{R_2}+\frac{R_{2}}{R_{N_v}}) \\
&\dots \\
&+(\frac{R_{N_v-1}}{R_{N_v}}+\frac{R_{N_v}}{R_{N_v-1}})
\label{eq58}
\end{aligned}
\end{equation}

According to the inequality properties, for any $j$ and $k$, $\frac{R_k}{R_j}+\frac{R_j}{R_k} \ge 2$. In the inequality, the equal sign is satisfied if and only if ${R_j=R_k}$. Therefore, according to Eq. \eqref{eq58}, we have $\sum_{j=1}^{N_v}\sum_{k=1,k \neq j}^{Nv}\frac{R_k}{R_j} \ge N_v(N_v-1)$. The equal sign is satisfied if and only if ${R_1=R_2=\dots=R_{N_v}}$. Substituting Eq. \eqref{eq30} into Eq. \eqref{eq29}, we have

\begin{equation}
R_j = \frac{2}{(W_0^j-1)\times T_{slot}}, \quad 1 \leq j \leq N_v,
\label{eq59}
\end{equation}

According to Eq. \eqref{eq59}, ${R_1=R_2=\dots=R_{N_v}}$ means ${W_0^1-1=W_0^2-1=\dots=W_0^{N_v}-1}$, i.e., $W_0^1=W_0^2=\dots=W_0^{N_v}$. Moreover $(W_0^1+W_0^2+\dots+W_0^{N_v})/N_v=\overline{W}_0$, thus the minimum age of information of the network is achieved when ${W_0^1 \approx W_0^2 \approx \dots \approx W_0^{N_v} \approx \overline{W}_0}$.


\ifCLASSOPTIONcaptionsoff
  \newpage
\fi


\end{document}